\pdfoutput=1

\documentclass[sigconf,usenames,dvipsnames]{acmart}

\usepackage{cleveref}
\usepackage{tabularx}
\usepackage{mathtools}
\usepackage{subcaption}
\usepackage{enumitem}
\usepackage{multirow} 
\usepackage{svg}

\DeclareMathOperator*{\argmin}{argmin}

\newcommand{\problem}{\textsc{Best Approximation Refinement}} %

\newcommand{\running}{\textit{scholarship query}}
\newcommand{\candidates}{\texttt{Students}}
\newcommand{\internships}{\texttt{Activities}}

\newcommand{\lb}[2]{\ell_{#1, #2}}
\newcommand{\ub}[2]{\mathscr{u}_{#1, #2}}

\newcommand{\constraints}{\mathcal{C}}

\newcommand{\conds}{\textsf{Preds}}
\newcommand{\cat}{\textsf{Cat}}
\newcommand{\num}{\textsf{Num}}
\newcommand{\attr}{\textsf{Attr}}
\newcommand{\sign}{\textsf{Sign}}

\newcommand{\prov}{\textsf{Lineage}}

\newcommand{\numbercircled}[1]{\textcircled{\small #1}}

\usepackage{subfiles}

\begin{document}

\title{Query Refinement for Diverse Top-$k$ Selection (Technical Report)}

\author{Felix S. Campbell}
\orcid{0000-0003-3888-1491}
\affiliation{%
  \institution{Ben-Gurion University of the Negev}
  \city{}
  \state{}
  \country{}
}
\email{felixsal@post.bgu.ac.il}

\author{Alon Silberstein}
\orcid{0009-0009-7591-1267}
\affiliation{%
  \institution{Ben-Gurion University of the Negev}
  \city{}
  \state{}
  \country{}
}
\email{alonzilb@post.bgu.ac.il}

\author{Julia Stoyanovich}
\orcid{0000-0002-1587-0450}
\affiliation{%
  \institution{New York University}
  \city{}
  \state{}
  \country{}
}
\email{stoyanovich@nyu.edu}

\author{Yuval Moskovitch}
\orcid{0000-0001-5109-3700}
\affiliation{%
  \institution{Ben-Gurion University of the Negev}
  \city{}
  \state{}
  \country{}
}
\email{yuvalmos@bgu.ac.il}

\begin{abstract}
Database queries are often used to select and rank items as decision support for many applications. As automated decision-making tools become more prevalent, there is a growing recognition of the need to diversify their outcomes. In this paper, we define and study the problem of modifying the selection conditions of an {\tt ORDER BY} query so that the result of the modified query 
closely fits some user-defined notion of diversity while simultaneously maintaining the intent of the original query. We show the hardness of this problem and propose a mixed-integer linear programming (MILP) based solution. We further present optimizations designed to enhance the scalability and applicability of the solution in real-life scenarios. 
We investigate the performance characteristics of our algorithm and show its efficiency and the usefulness of our optimizations.

\end{abstract}


\maketitle

\section{Introduction}
\label{sec:introduction}

Ranking-based decision making is prevalent in various application domains, including hiring~\cite{GeyikAK19} and school admission~\cite{peskun2007effectiveness}. Typically, this process involves selecting qualifying candidates based on specific criteria (e.g., for a job position) and ranking them using a quantitative measure to identify the top candidates among those who qualify (e.g., for a job interview or offer). This process may be automated and expressed using %
SQL queries, with the {\tt WHERE} clause used to select candidates who meet certain requirements, and the {\tt ORDER BY} clause used to rank them. We next illustrate this idea using a simple example in the context of awarding scholarships.

\begin{example}
	\label{ex:running}
	Consider a foundation that wishes to grant six high-performing students scholarships to universities in order to encourage participation in STEM programs. 
 The foundation utilizes a database of all students seeking scholarships provided by their schools, which may be filtered according to the requirements of the foundation. 
 \Cref{tab:candidates} shows the students dataset, consisting of five attributes: a unique ID, gender, family's income level, grade point average (GPA), and SAT score. The schools also provide information on the student's 
 involvements with extracurricular activities, which are shown in \Cref{tab:internships} as a dataset with two attributes: the student's ID and an abbreviation representing the activity in which they participated. The set of activities in \Cref{tab:internships} consists of robotics ($RB$), Science Olympiad ($SO$), Math Olympiad ($MO$), game development ($GD$), and a STEM tutoring organization ($TU$).
 
    The foundation would like to award these scholarships to students who have displayed interest in STEM fields through their involvement in extracurricular activities and have maintained a minimum GPA. The selected students are ranked by their SAT exam scores, and the foundation grants funding to the best six students and additional funding to the top three students.
    These requirements can be expressed using the following query, which selected students who have participated in an extracurricular robotics club with a minimum GPA of $3.7$:
    
    \begin{center}
    \footnotesize
    \begin{tabular}{l}
        \verb"SELECT DISTINCT ID, Gender, Income "\\
        \verb"FROM Students NATURAL JOIN Activities"\\
        \verb"WHERE GPA >= 3.7 AND Activity = 'RB'"\\
        \verb"ORDER BY SAT DESC"\\
    \end{tabular}
    \end{center}
    We refer to this query throughout as the \running{}. Evaluating this query over the datasets in \Cref{tab:candidates,tab:internships} produces the ranking $[t_4, t_7, t_8, t_{10}, t_{11}, t_{12}]$ %
    Therefore, the foundation awards students $t_4$, $t_7$, $t_8$ with an extra scholarship and students $t_{10}$, $t_{11}$, and $t_{12}$ with the regular amount.
\end{example}

\begin{table}[ht]
\begin{minipage}{0.6\linewidth}
        \centering
		\caption{\candidates{}}
		\label{tab:candidates}
        \footnotesize
		\begin{tabular}{lccrr}
		\hline
		\textbf{ID} & \textbf{Gender} & \textbf{Income} & \textbf{GPA} & \textbf{SAT $\downarrow$} \\ \hline
		$t_1$           & M               & Medium         & 3.7      & 1590 \\    
		$t_2$           & F               & Low         & 3.8      & 1580 \\    
		$t_3$           & F               & Low         & 3.6      & 1570 \\    
		$t_4$           & M               & High         & 3.8      & 1560 \\    
		$t_5$           & F               & Medium         & 3.6      & 1550 \\    
		$t_6$           & F               & Low         & 3.7      & 1550 \\    
		$t_7$           & M               & Low         & 3.7      & 1540 \\    
		$t_8$           & F               & High         & 3.9      & 1530 \\    
		$t_9$           & F               & Medium         & 3.8      & 1530 \\    
		$t_{10}$           & M               & High         & 3.7      & 1520 \\    
		$t_{11}$           & F               & Low         & 3.8      & 1490 \\    
		$t_{12}$           & M               & Medium         & 4.0      & 1480 \\    
		$t_{13}$           & M               & High         & 3.5      & 1430 \\  
  	$t_{14}$           & F               & Low         & 3.7      & 1410 \\  
		\end{tabular}
\end{minipage}
\hfill
\begin{minipage}{.35\linewidth}
    \centering
    \caption{\internships{}}
    \label{tab:internships}
    \footnotesize
    \begin{tabular}{lc}
    \hline
    \textbf{ID} & \textbf{Activity} \\ \hline
        $t_{1}$            & SO        \\  
        $t_{2}$            & SO        \\  
        $t_{3}$            & GD        \\  
        $t_{4}$            & RB        \\
        $t_{4}$            & TU        \\
        $t_{5}$            & MO        \\  
        $t_{6}$            & SO        \\  
        $t_{7}$            & RB        \\  
        $t_{8}$          & RB        \\
        $t_{8}$           & TU        \\
        $t_{10}$           & RB        \\  
        $t_{11}$           & RB        \\  
        $t_{12}$           & RB        \\
        $t_{14}$          & RB        \\
    \end{tabular}
\end{minipage}
\end{table}

If the query is part of some high-stake decision-making process, stating diversity requirements as cardinality constraints over the presence of some demographic groups in the top-$k$ result is natural. For instance, in the above example, the foundation may wish to promote female students in STEM by awarding a proportional number of scholarships to male and female applicants, i.e.,  top-$6$ tuples in the output should include at least three females. Moreover, to expand access to STEM education, the foundation may also wish to limit the extended scholarships granted to students from high-income families. Namely, the top-$3$ results should include at most one student with a high income. The \running{} does not satisfy these constraints since the top-$6$ tuples are $t_4, t_7, t_8, t_{10}, t_{11}$ and $t_{12}$ include only two females ($t_8$ and $t_{11}$), and the top-$3$ includes two students from high-income families ($t_4$ and $t_{8}$).

 In this paper, we propose a novel \emph{in-processing} method to improve the diversity of a ranking by refining the query that produces it.%
 
 \begin{example}
    \label{ex:first_refinement}
    The \running{} may be refined by adjusting the condition on \verb"Activity" to include students involved in Science Olympiad ($SO$), resulting in the following query:
    \begin{center}
    \footnotesize
    \begin{tabular}{l}
        \verb"SELECT DISTINCT ID, Gender, Income "\\
        \verb"FROM Students NATURAL JOIN Activities"\\
        \verb"WHERE GPA >= 3.7 AND (Activity = 'RB' OR Activity = 'SO')"\\
        \verb"ORDER BY SAT DESC"\\
    \end{tabular}
    \end{center}
    Note that the essence of the query (selecting students who have displayed interest in STEM) is maintained by the refined query, while the constraints are satisfied as the top-$6$ tuples ($t_1, t_2, t_4, t_6, t_7,$ and $t_8$) consist of three women ($t_2, t_6$ and $t_8$) where the top-$3$ includes only a single student ($t_4$) with high-income.

\end{example}

The notion of refining queries to satisfy a set of diversity constraints was recently presented in~\cite{ERICA, ERICAfull}, however, this work focuses on cardinality constraints over the {\it entire} output and does not consider the order of tuples.
The problem of ensuring diverse outputs in ranking queries has received much recent attention from the research community~\cite{YS17, CSV18, YGS19, AJS19, KR18, CMV20}. For instance, in \cite{YS17, CSV18, YGS19}, output rankings are modified directly in a \emph{post-processing step}, in order to satisfy a given set of constraints over the cardinality of protected groups in the ranking. E.g., to fulfill the desired constraints in the above example, the foundation may manipulate the output, awarding $t_4$, $t_5$, $t_6$, $t_7$, $t_8$ and $t_{10}$ with a scholarship, where $t_4$, $t_5$ and $t_6$ will get the extended grant. However, this leaves open the question of \emph{how one may obtain such a ranking in the first place}.

Further, post-processing may be problematic to use to improve diversity, for two reasons: (1) by definition, it modifies the results after they were computed, raising a procedural fairness concern, and (2) it may explicitly use information about demographic or otherwise protected group membership, raising a disparate treatment concern. In contrast, in-processing %
is usually legally permissible, essentially because it applies the same evaluation process to all individuals. That is, by modifying the query we produce a new set of requirements, and test all individuals against these same requirements.  In contrast, post-processing methods may decide to include or exclude individuals based on which groups they belong to, therefore treading individuals differently depending on group membership.
Alternative in-processing solutions involve adjustments to the ranking algorithm~\cite{AJS19} or modifying items to produce a different score~\cite{KR18, CMV20}. Our approach, conversely, assumes the ranking algorithms and scores of different items are well-designed and fixed, and we aim to modify the set of tuples to be ranked.

 Our goal is to find minimal refinements to the original query that fulfill a specified set of constraints, however, we note that the notion of minimality may be defined in different ways, depending, for example, on the legal requirements or on the user's preferences.

 \begin{example}
    \label{ex:refine-and-distance}
    We may refine the \running{} by relaxing the GPA requirement to $3.6$ and including students who participated in a game development activity $(GD)$,  obtaining the following query: 
    \begin{center}
    \footnotesize
    \begin{tabular}{l}
        \verb"SELECT DISTINCT ID, Gender, Income "\\
        \verb"FROM Students NATURAL JOIN Activities"\\
        \verb"WHERE GPA >= 3.6 AND (Activity = 'RB' OR Activity = 'GD')"\\
        \verb"ORDER BY SAT DESC"\\
    \end{tabular}
    \end{center}
    Similarly to the refined query from Example~\ref{ex:first_refinement}, the top-$6$ students %
    ($t_3$, $t_4$, $t_7$, $t_8$, $t_{10}$, $t_{11}$, and $t_{12}$) include three women ($t_3$, $t_8$, and $t_{11}$), and there is only a single high-income student ($t_3$) among the top-$3$. While the predicates of this refined query are intuitively more distant from the original query than our prior refinement in Example~\ref{ex:first_refinement} (two modifications compared to a single one), its output is more similar to the output of the original query (the top-$3$ sets differ by one tuple).
    
\end{example}

To accommodate different alternative query relaxation objectives, as illustrated above, we propose a framework that allows the user to specify their preferred notion of minimality. 

To the best of our knowledge, our work is the first to intervene on the ranking process by modifying {\it which} items are being considered by the ranking algorithm. This admits a large class of ranking algorithms while keeping the relative order of tuples consistent. This---of course---does not come for free; the coarseness of refinements means that there may be no refinement that produces a satisfactorily diverse ranking. Therefore, we study the problem of finding a refined query %
that is within a specified maximum distance from satisfying all of the constraints, if one exists. %

\paragraph*{Contributions \& roadmap}
We begin by formalizing the \problem{} problem of obtaining a refined query that is closest, according to a given distance measure, to the original query, while still adhering to a set of cardinality constraints within a maximum distance, and show that this problem is {\sf NP-hard} (\Cref{sec:background}).
We thus propose modeling the problem as a Mixed Integer Linear Program (MILP) and utilizing it to derive an approximate solution (\Cref{sec:search}). Inspired by the use of data annotations (provenance) to perform what-if analysis, i.e., to efficiently reevaluate queries using algebraic expression without constantly accessing a DBMS (see, e.g., ~\cite{DIMT13, BourhisDM16, DeutchMT14, MoskovitchLJ22}), we construct the MILP using data annotation variables. 
This formulation offers two advantages: it allows us to leverage the effectiveness of existing MILP solvers while avoiding the costly reevaluation of queries on the DBMS.

Existing MILP solvers may solve the problem efficiently however their performance is sensitive to the size of the program. The program generated in Section~\ref{sec:search} is linear in the data size, which can be challenging in real-life scenarios as we demonstrate in our experimental evaluation. We, therefore, propose optimizations to make our approach more scalable, using the relevancy of the data and the structure of the set of cardinality constraints to prune and relax our problem (\Cref{sec:optimizations}). In \Cref{sec:experiments}, we present an extensive experimental evaluation. We developed a dedicated benchmark consisting of real-life datasets and considering realistic scenarios. Our results show the efficiency and scalability of our approach with respect to different parameters of the problem.

\section{Problem Overview}
\label{sec:background}

In this paper, we consider the class of conjunctive Select\footnote{{\tt DISTINCT} is supported, and is used to select individuals uniquely if they should not appear more than once in the output.}-Project-Join (SPJ)\footnote{We note that the system may be extended easily to handle unions, but we omit its description here due to space constraints.} queries with an \texttt{ORDER BY}~$s$ clause, generating a ranked list of tuples, where $s$ is a \emph{score function} of a single tuple $t$.  
A query $Q$ may have numerical and categorical selection predicates, denoted $\num(Q)$ and $\cat(Q)$, respectively. Numerical predicates are of the form $A \diamond C$, where $A$ is a numerical attribute $C \in \mathbb{R}$, and $\diamond \in \{<, \leq,=, >, \geq \}$. %
Categorical predicates are of the form $\bigvee_{c \in C} A = c$, where $A$ is a categorical attribute and $C$ is a set of constants from the domain of $A$. Selection operators combine predicates by taking their conjunction. We use $\conds(Q)$ to denote the set of attributes appearing in the selection predicates of $Q.$
In the rest of the paper, we simply use query to refer to such queries.

\subsection{Preliminaries}
\label{sec:cardinality-constraints}
\paragraph*{Cardinality constraints} Imposing constraints on the cardinality of tuples belonging to a certain group in a query result to mitigate bias and improve diversity was studied in~\cite{MLJ22, ERICA}. 
In the context of ranking, cardinality constraints are used over the top-$k$ of the ranking for various values of $k$ (see, e.g., \cite{CSV18,YGS19, ICDE23}). Following this vein of research, we allow users to define constraints on the cardinality of groups (i.e., data subgroups) for this setting. 

A {\it group} is a collection of tuples that share the same value(s) for one or more (categorical) attributes and is defined by a conjunction of conditions over values of the attributes. For instance, in \Cref{ex:running}, the group including women students 
is defined by the condition ${\tt Gender} = F$ and consists of students $t_2$, $t_3$, $t_5$, $t_6$, $t_8$,  $t_9$,  $t_{11}$, and $t_{14}$. The group of low-income women candidates is then defined by the condition ${\tt Gender} = F \wedge {\tt Income} = Low$ and consists of students $t_2$, $t_3$, $t_6$, $t_{11}$, and $t_{14}$. A cardinality constraint $\lb{G}{k} = n$ (or $\ub{G}{k} = n$) specifies a lower (or an upper) bound of $n$ tuples belonging to a group $G$ appearing within the top-$k$ tuples of the result. For instance, in our running example, the constraint ``at least $3$ of the top-$6$ candidates are women'', can be expressed as  $\lb{{\tt Gender} = F}{k=6} = 3$. Multiple cardinality constraints may be composed together, forming a constraint set that we denote by~$\constraints{}$.

\paragraph*{Refinements} We use the notion of query refinement defined in \cite{MK09}. Given a query $Q$, a refinement of $Q$ modifies its selection predicates. A numerical predicate $A \diamond C~\in \num(Q)$ is a modification to the value of $C$. For categorical predicates $\bigvee_{c \in C} A = c~ \in \cat(Q)$, a refinement is done by adding and/or removing predicates from the set of values $C$.  We say that a query $Q'$ is a \emph{refinement} of query $Q$ if $Q'$ is obtained from $Q$ by refining some predicates of $Q$.

\begin{example}
    The \running{} has two predicates: a numerical predicate {\tt GPA $\geq$ 3.7} and a categorical predicate {\tt Activity = `RB'}. A possible refinement of the numerical predicate may be {\tt GPA $\geq$ 3.6}. The categorical predicate may be refined by adding {\tt `GD'} to $C$. The refined query resulting from $Q$ by applying these refinements is the query depicted in Example~\ref{ex:refine-and-distance}. 
\end{example}

\subsection{Refinement Distance}
\label{sec:distance}

Our objective is to find a refinement $Q'$ that fulfills a specified set of constraints \emph{and} preserves the essence of the intent of query $Q$, i.e., is in some sense close to $Q$. A key question is how to measure the distance between a query $Q$ and a refinement $Q'$. %

Recall from \Cref{ex:refine-and-distance} that there may be multiple ways to define such distance. %
In this paper, we support distance functions of two kinds --- those that compare the predicates of $Q$ and $Q'$ (\emph{predicate-based}) and those that compare the top-$k$ results of $Q$ and $Q'$, either as sets or in ranked order (\emph{outcome-based}).  In both cases, a distance function returns a real number, with a smaller value indicating closer proximity between $Q$ and $Q'$.
As we will discuss later, we use mixed-integer linear programming to find query refinements. Hence, the distance function must be linear (or able to be linearized) in the variables of its input. However, this limitation still permits a diverse set of valuable distance measures, as we demonstrate next.

\paragraph*{Predicate-based distance}\sloppy
    Given a query $Q$ and a refinement $Q'$, a natural distance measure with respect to a numerical predicate $n_Q = A \diamond C~\in \num(Q)$ is $\vert n_Q.C - n_{Q'}.C \vert$, where $n_Q.C$ is the value of $C$ in $n_Q$ and $n_{Q'}.C$ is the value of $C$ in $Q'$. The distance between all numerical predicates may be (normalized and) aggregated as $\sum_{n_Q\in \num(Q)} \frac{\vert n_Q.C - n_{Q'}.C \vert}{n_Q.C}$.
     The distance between categorical attributes may be measured using the Jaccard distance, defined for a pair of sets $R$ and $S$ as $J(R, S) = 1 - \frac{|R \cap S|}{|R \cup S|}$. We may aggregate the distance across categorical predicates as $\sum_{c_Q \in \cat(Q)} J(c_Q.C, c_{Q'}.C)$, where $c_{Q'}\in \cat(Q')$ is the corresponding categorical attribute of $c_Q\in \cat(Q)$ ($c_Q$ and $c_{Q'}$ are of the form $\bigvee_{c \in C} A = c$). 
     
     Combining numerical and categorical components, we formulate the predicate-based distance between $Q$ and $Q'$ as: %
    $$DIS_{pred}(Q,Q')=\sum_{n_Q\in \num(Q)} \frac{\vert n_Q.C - n_{Q'}.C \vert}{n_Q.C}+ \sum_{c_Q \in \cat(Q)} J(c_Q.C, c_{Q'}.C)$$

\begin{example}
    Let $Q'$ and $Q''$ be the refinements of the \running{} $Q$ shown in Examples~\ref{ex:first_refinement} and ~\ref{ex:refine-and-distance}, respectively. We compute $DIS_{pred}(Q,Q') = \frac{3.7 - 3.7}{3.7} + (1 - \frac{|\{RB\}|}{|\{RB, SO\}|}) = 0.5$, which is smaller than $DIS_{pred}(Q,Q'') = \frac{3.7 - 3.6}{3.7} + (1 - \frac{|\{RB\}|}{|\{RB, MO\}|}) \approx 0.53$.
\end{example}    
\paragraph*{Outcome-based distance} 
Distance measures in this family compare the top-$k$ items $Q(D)_k$ and $Q'(D)_k$, for some value of $k$.  %
We consider two types of distance measures: those that look at the top-$k$ as sets, and those that are sensitive to the ranked order among the top-$k$ items.  We give a couple of examples below, noting that many other set-wise and rank-aware distance metrics can be defined.

A natural distance metric computes the Jaccard distance between the \emph{sets} of top-$k$ items of $Q$ and $Q'$: $DIS_{Jaccard}(Q, Q', k) = J(Q(D)_k, Q'(D)_k)$.
\begin{example}
Let $Q'$ and $Q''$, again, be the refinements of the \running{} $Q$ shown in Examples~\ref{ex:first_refinement} and ~\ref{ex:refine-and-distance}, respectively.
Then, $DIS_{Jaccard}(Q, Q',k=3) = 1 - \frac{|\{t_4\}|}{|\{t_1, t_2, t_4, t_7, t_8\}|} = 0.8$, while $DIS_{Jaccard}(Q, Q'', k=3) = 1 - \frac{|\{t_4, t_7\}|}{|\{t_3, t_4, t_7, t_8\}|} = 0.5$. 
\end{example}

Observe that $Q''$ is closer to $Q$ according to $DIS_{Jaccard}$ at top-$3$, while $Q'$ is closer to $Q$ according to $DIS_{pred}$.

Recall that query refinement does not reorder tuples.  That is, tuples that belong to $Q(D)_k \cap Q'(D)_k$, will appear in the same relative order in both top-$k$ lists. As another alternative, a rank-aware measure may, for example, use a variant of Kendall's $\tau$~\cite{10.1093/biomet/30.1-2.81} that was proposed by~\citet{FKS03} to compare the top-$k$ items of $Q$ and $Q'$. In a nutshell, this measure, which we'll denote $DIS_{Kendall}(Q, Q',k)$, considers the new tuples in the top-$k$ (i.e.,  $Q'(D)_k \setminus Q(D)_k$), and computes how much the tuples in the original top-$k$ ($Q(D)_k$) were displaced. (Cases 2 and 3 from~\cite{FKS03} apply in our setting.) %

Intuitively, if $DIS_{Kendall}(Q, Q',k=3) < DIS_{Kendall}(Q, Q'',k=3)$, then 
the tuples $Q''(D)_k \setminus Q(D)_k$
are positioned closer to the top of the list than those in $Q'(D)_k \setminus Q(D)_k$.

 \begin{example}
     To illustrate $DIS_{Kendall}$, we introduce a new refinement $Q'''$, which we define as:
     \begin{center}
     \footnotesize
     \begin{tabular}{l}
        \verb"SELECT DISTINCT ID, Gender, Income "\\
        \verb"FROM Students NATURAL JOIN Activities"\\
        \verb"WHERE GPA >= 3.6 AND (Activity = 'CS' OR Activity = 'MO')"\\
        \verb"ORDER BY SAT DESC "\\
     \end{tabular}
     \end{center}
     Observe that $DIS_{pred}(Q,Q'') = DIS_{pred}(Q,Q''')$ and $DIS_{Jaccard}(Q,Q'',k=3) = DIS_{Jaccard}(Q,Q''',k=3)$ 
     However, the resulting ranking of $Q'''(D)$ is $[t_4, t_5, t_7, t_8, t_{10}, t_{11}, t_{12}]$. This refinement includes a new tuple $t_5$ in the output where it ranks second, while in $Q''$, the new tuple included ($t_3$) ranks first in $Q''(D)$. However, we find that $DIS_{Kendall}(Q,Q'', k=3) > DIS_{Kendall}(Q,Q''',k=3)$,
     meaning that $Q'''$ is preferable to $Q''$ according to this measure. %
 \end{example}

    These measures can be combined to formulate new measures that take into account both the queries' predicate distance and the outputs, e.g., using a weighted function.  %

\subsection{Problem Formulation}
Given a query $Q$, a set of cardinality constraints $\constraints$, and a distance measure, our goal is to find a refinement with minimal distance from $Q$ that satisfies the set of constraints. However, we can show that such a refinement may not exist.

\begin{table}%
    \caption{Relation used for proof of \Cref{thm:no-perfect}}
    \label{tab:no-perfect}
    \small
    \footnotesize
    \begin{tabular}{lll}
    \hline
    \textbf{X} & \textbf{Y} & \textbf{Z} \\ \hline
            A           & C               & 6          \\
            A           & D               & 5           \\
            A           & D               & 4           \\
            B           & C               & 3          \\
            A           & C               & 2          \\
            B           & D               & 1          
    \end{tabular}
\end{table}

\begin{theorem}
	\label{thm:no-perfect}
	There exists a dataset $D$, a query $Q$ over $D$, and a constraints set $\constraints$ such that no refinement of $Q$ satisfies $\constraints$.
\end{theorem}
\begin{proof}
	We prove this claim by a simple example. Let $Q$ be the query {\tt SELECT * FROM "\Cref{tab:no-perfect}" WHERE Y = 'C' OR Y = 'D' ORDER BY Z DESC}. Let us require that $2$ tuples from group {\tt X ='B'} (or just $B$ for brevity) appear in the top-$3$ of the ranking, i.e., setting $\lb{X='B'}{k=3} = 2$. The original query evaluated over \Cref{tab:no-perfect} selects the entire table, resulting in a ranking with no tuples belonging to $B$ in the top-$3$. There are then only two possible refinements on the original query: {\tt Y = 'C'} or {\tt Y = 'D'}. In both cases, there is only $1$ item of $B$ in the top-$3$.
    Neither the original query nor any of its possible refinements result in a query that satisfies the constraints.%
\end{proof}

\Cref{thm:no-perfect} motivates the need to find a refinement that deviates as little as possible from satisfying the constraint set in the case that exact constraint satisfaction is impossible, which allows us to provide results that are more useful to the user than simply stating its infeasibility.
To measure the deviation from the satisfaction of a given set of constraints $\constraints$, we leverage the notion of the \emph{mean absolute percentage error}, as was done in~\cite{BAM14}. Specifically, we use it to measure the deviation from the constraints over groups in $\constraints$ and their cardinalities in the output of the (refined) query. We modify its definition to not penalize rankings that are above (below) the cardinalities specified in lower (upper) bound constraints for a group. 
\begin{definition}[Deviation]%
\label{def:mospe}
Recall that $Q(D)_k$ denotes the top-$k$ tuples in the output of the query $Q$ over a database $D$.  The deviation between $\constraints$ and $Q$, $DEV(Q(D), \constraints)$ is given by
\begin{align*}
 \frac{1}{|\constraints|} \sum_{(\mathscr{c}_{G, k} = n) \in \constraints} \frac{\max \left(\sign(\mathscr{c}) \cdot (n - |Q(D)_k \cap G|), 0 \right)}{n} 
\end{align*}
where $\sign(\mathscr{c})$ is $1$ for lower-bound constraints ($\ell$) and $-1$ for upper-bound constraints ($\mathscr{u}$). Larger values represent a larger violation of the constraint set.
\end{definition}

When computing deviation, we assume that the output of $Q(D)$ has at least the number of tuples of the largest $k$ with a constraint in $\constraints$. We refer to this quantity throughout as $k^*$.
We are now ready to formally define the \problem{} problem.

\begin{definition}[\problem{}]
\label{def:best-approx-refinement}
    Given a database $D$, a query $Q$, a constraint set $\constraints$, a maximum deviation from the constraint set $0 \leq \varepsilon$ and a distance measure $DIS : Q \times \mathcal{R} \times k \rightarrow \mathbb{R}$, the answer to the \problem{} problem is the refinement $Q'$ in
    \begin{align*}
        \argmin_{Q' \in \mathcal{R}}\ DIS(Q, Q', k) \ \ \textrm{such that} \ \ & DEV(Q'(D), \constraints) \leq \varepsilon
    \end{align*}
    where $\mathcal{R}$ is the set of possible refinements of $Q$ that have at least $k^*$ tuples in in their output. %
    Note that the $k$ parameter is optional in the distance measure (e.g., $DIS_{pred}$ does not include it). %
    A special value is returned if there is no refinement $Q'$ with constraint set deviation at most $\varepsilon$. %
\end{definition}

\problem{} provides the most similar (according to the given similarity definition) refinement with an acceptable deviation from satisfying the constraint set.

\subsection{Hardness}

We can show that this problem is {\sf NP-hard}.

\begin{theorem}
\label{thm:hardness}
    \problem{} is {\sf NP-hard}.
\end{theorem}

The proof is based on a reduction from {\sc Vertex-Cover}, a well-known {\sf NP-complete} decision problem~\cite{K72}. To this end, we define the following corresponding decision problem. Given a database $D$, a query $Q$, a constraint set $\constraints$, a maximum deviation from the constraint set $0 \leq \varepsilon$, a value $k$, a distance measure $DIS : Q \times \mathcal{R} \times k \rightarrow \mathbb{R}$ and a maximum distance $\delta \geq 0$, determine whether there exists a refinement $Q'\in \mathbb{R}$ such that $DEV(Q'(D), \constraints) \leq \varepsilon$ and $DIS(Q, Q', k)~\leq~\delta$. 

An input to the {\sc Vertex-Cover} problem consists of an undirected graph $G = (V, E)$ and a number $S$, and the goal is to determine whether there exists a {\it vertex cover}, i.e., a subset of vertices $V' \subseteq V$ such that for every edge $(u, v)$ in $E$, one or both of its endpoints, $u$ and $v$, are in $V'$ and $|V'| \leq S$.

\begin{proof}
    We reduce {\sc Vertex-Cover} to the stated decision version of our problem. Let $G = (V, E)$ be the graph for which we want to decide whether a vertex cover of at most $S$ vertices exists. In order to construct an instance for our decision problem, we first construct a database $D$ with a relation {\tt Edges} with categorical attributes {\tt Edge} and {\tt Vertex}, and a numerical attribute {\tt IsDummy}. For each edge $(u, v)$ in $E$, we label it uniquely with an integer $i$ from $[1\ ..\ |E|]$ and insert into {\tt Edges} two tuples: $(i, u, 0)$ and $(i, v, 0)$. In order to handle some caveats, we also include a dummy vertex $\star$ which does not exist in $V$, and create a tuple $(|E| + i, \star, 1)$ for all $i$ in $[1\ .. \ 2 \cdot |E|]$ which we will call dummy edges. Then, let the query $Q$ that we refine be {\tt SELECT * FROM Edges WHERE Vertex = $\star$ ORDER BY IsDummy ASC}. Its output is a ranking consisting of any of the tuples generated by the edges in $E$ with an endpoint in any of the vertices in the {\tt Vertex} predicate (or the dummy edges not in $E$), which before refinement consists solely of $\star$. We then build a constraint set $\constraints$ with the constraints $\lb{Edge = e}{k=2 \cdot |E|} = 1$ for each edge label $e$ in $[1\ ..\ |E|]$. Such constraints are perfectly satisfied when at least one of the tuples generated by its associated edge is in the ranking. Note that our constraints are over the top-$k$ where $k$ is such that it includes all the data in {\tt Edges} (of size $2 \cdot |E|$), as all real edges are guaranteed to be in the top-$(2 \cdot |E|)$ due to the {\tt ORDER BY} clause of the query. Then, there exists a vertex cover of size at most $S$ for $G$ if and only if there exists a refinement $r$ for $Q$ over $D$ with (input) distance of at most $1 - \frac{1}{S + 1}$ and deviation from $\constraints$ of at most zero. We now show both directions of our equivalence. \\
    \indent \underline{Vertex cover $\Rightarrow$ refinement:} Let $V'$ be a vertex cover for $G$ of size at most $S$. Then, we may refine the {\tt Vertex} predicate of $Q$ to include all vertices of $V'$ and the dummy vertex $\star$ (if not included, the intersection in the Jaccard distance would be empty, making the distance $1$). This refinement $r$ has input distance $1 - \frac{|\{ \star \} \cap (V' \cup \{ \star \})|}{|\{ \star \} \cup (V' \cup \{ \star \})|} = 1 - \frac{1}{|V'| + 1}$ which is at most $1 - \frac{1}{S + 1}$. Now $r$ is a refinement over $Q$ that selects tuples generated by edges in $E$ with endpoints in $V'$. The resulting ranking $r(D)$ must contain at least one of the tuples generated by each edge in $E$ as we suppose that $V'$ is a vertex cover, i.e. every edge in $E$ has an endpoint in $V'$. The resulting ranking also has at least $2 \cdot |E|$ tuples due to the dummy edges selected by including the dummy vertex (an assumption required for calculating deviation over $r(D)$). Even though the dummy vertex is selected, all the dummy edges will appear {\it after} the real edges, making it so that the real edges will always be in the first $2 \cdot |E|$ tuples. Therefore, as all edges have at least one of their tuples selected and they are guaranteed to be in the top-$(2 \cdot |E|)$, the resulting ranking has that $1 \leq |\sigma_{Edge=e}(r(D)_{2\cdot|E|})| \leq 2$ for each edge label $e$ in $[1\ ..\ |E|]$. Since this is the case, all constraints are then perfectly satisfied, resulting in zero deviation. Therefore, if there is a vertex cover with size at most $S$ for $G$, then there exists a refinement over our query $Q$ and constraint set $\constraints$ such that the input distance and deviation are within our stated bounds. \\
    \indent \underline{Refinement $\Rightarrow$ vertex cover:} Let $r$ be a refinement over $Q$ and $D$ with a deviation from $\constraints$ of at most zero and a distance of at most $1 - \frac{1}{S + 1}$. Note that the vertices included in the {\tt Vertex} predicate of the refined query $r$ must include the dummy vertex $\star$, otherwise the refinement's input distance would always be 1 for the same reason mentioned above. This dummy vertex may always be safely excluded after finding the refinement, as it does not contribute anything meaningful to the output of $r(D)$ as no edges in $E$ have it as an endpoint. Therefore, we denote $V'$ as the vertices included in the {\tt Vertex} predicate of $r$ without $\star$. If a refinement exists with deviation at most zero then all constraints were perfectly satisfied. Therefore, for each edge label $e$ in $[1\ ..\ |E|]$ we have that $1 \leq |\sigma_{Edge=e}(r(D)_{2\cdot|E|})| \leq 2$. This means all edges in $E$ have one of their tuples in the output ranking of $r(D)$, i.e. all edges in $E$ have an endpoint in $V'$, making $V'$ a vertex cover for $G$. Now, we must verify that $V'$ consists of at most $S$ vertices. Because our refinement has distance of at most $1 - \frac{1}{S + 1}$, the size of $V' \cup \{ \star \}$ is bounded by $S + 1$, making the size of $V'$ bounded by $S$. This means that $V'$ contains at most $S$ vertices, and we have already shown that every edge in $E$ has an endpoint in $V'$ by the satisfaction of the constraint set. Therefore, if there exists a refinement over $Q$ and $D$ with (input) distance at most $1 - \frac{1}{S + 1}$ and deviation at most zero, then there exists a vertex cover of at most $S$ vertices for $G$.  
    
    \indent Therefore, we can solve {\sc Vertex-Cover} with an algorithm that solves our decision problem, making it at least as hard as {\sc Vertex-Cover}, which is {\sf NP-complete}.
\end{proof}

\section{Finding The Best Approximation}
\label{sec:search}

Our problem may be solved na\"{\i}vely by an exhaustive search over the possible refinements. However, the search space of refinements becomes intractably large even with relatively modest datasets, as the number of possible refinements is exponential in the number of the query's attributes.
Beyond the high cost of an exhaustive search, a na\"{\i}ve solution would require the evaluation of each refinement query on the DBMS to check its deviation from the constraint set. 

To address these challenges, we propose a solution based on a mixed-integer linear program (MILP). Mixed-integer linear programming is a model for optimizing a linear objective function subject to a set of expressions (equalities and inequalities) linear in the discrete or continuous variables of the problem, limiting the space of feasible assignments. Solvers for such programs have been developed with techniques to solve even large problems efficiently in practice, as discussed in \cite{QFix}.
By incorporating the concepts introduced in~\cite{MLJ22, ERICA}, we utilize data annotations to depict potential refinements. These annotations serve as variables in the MILP, and enable us to quantify the deviation from the constraint set without having to reevaluate refinements across the DBMS.

Briefly, a solution for a MILP is an assignment for the variables in the expressions, such that they are satisfied and the objective function is minimized. 
Intuitively, given a database $D$, a query $Q$, a constraint set $\constraints$, a maximum deviation from the constraint set $0 \leq \varepsilon$, and a distance measure $DIS$, we construct an instance of MILP such that the solution corresponds to a minimal refinement that produces a ranking such that its deviation is within the maximum deviation $\varepsilon$ from the constraint set $\constraints$ while minimizing $DIS$.
By formulating \problem{} as a MILP, we can leverage existing tools to streamline the search process.

It is important to note that by using MILP to represent the problem, we are limited to distance measures that can be modeled by a linear program. However, this limitation still allows a wide range of useful distance measures, including the ones defined in~\Cref{sec:distance}. Some of these measures may require additional modeling techniques to become linearized. For example, when modeling the Jaccard distance, we can use the Charnes-Cooper transformation~\cite{CC62}. Similarly, we can introduce auxiliary binary variables to model the version of Kendall's $\tau$ for top-$k$ lists introduced in \cite{FKS03}. %

\begin{table}[t!]
    \centering
    \caption{Summary of variables used in our MILP model}
    \footnotesize
    \begin{tabularx}{0.48\textwidth}{c|c|X}
    {\bf Var.} & {\bf Domain} & {\bf Description} \\
    \hline
    $C_{A,\diamond}$  & $\mathbb{R}$ & Refined $C$ for a num. predicate on $A$ with operator $\diamond$                     \\
    $A_v$             & $\{0, 1\}$ & Whether a value $v$ is selected by the cat. predicate on $A$                           \\
    $A_{v, \diamond}$ & $\{0, 1\}$ &  Whether a value $v$ is in the range of the num. predicate on $A$ with operator $\diamond$ \\
    $r_t$            & $\{0, 1\}$ & Whether tuple $t$ is selected by the refinement                                     \\
    $s_{t}$          & $\mathbb{R}$ & Rank of tuple $t$ in the ranking generated by the refinement                              \\
    $l_{t, k}$       & $\{0, 1\}$ & Whether tuple $t$ is present in the top-$k$ of the ranking generated by the refinement          \\
    $E_{G, k}$       & $\mathbb{R}$ & Number of tuples to add (remove) to satisfy lower-bound (upper-bound) cardinality constraint    
    \end{tabularx}
    \label{tab:variables}
\end{table}

The MILP instance we construct consists of two main groups of expressions: those that require that all tuples selected by the refinement are in the ranking according to the {\tt ORDER BY} expression of $Q$, and those that enforce that the derived ranking's deviation from the constraint set does not exceed the input bound $\varepsilon$. We next explain the construction of the expressions in each set, and, in order to give a more intuitive picture of the process, demonstrate in \Cref{fig:example_summary} how variables are generated from a running example and how they are combined by these expressions.

\subsection{Modeling refinements output using expressions}

Inspired by the use of provenance for query refinements~\cite{MLJ22,ERICA}, we utilize the notion of data annotations to model refinements through a set of expressions. This set is divided into two parts. The first part is used to model the tuples that satisfy the refinement query's predicate, while the second part ensures that the selected tuples are ordered correctly by the {\tt ORDER BY} expression of the input query. We start by describing the variables used in the expressions.

\begin{table}%
        \centering	
  \caption{$\widetilde{Q}$ obtained from \running{}}
  \hspace*{-1cm}
        \footnotesize
		\begin{tabular}{clccc}
		\cline{2-5}
		& \textbf{ID} & \textbf{Gender} & \textbf{Income} &  $\prov(t)$ \\ \cline{2-5}
		 & $t_1$           & M               & Medium         & $\{Activity_{SO}, GPA_{3.7}\}$\\    
		 & $t_2$           & F               & Low         & $\{Activity_{SO}, GPA_{3.8}\}$\\    
		 & $t_3$           & F               & Low        & $\{Activity_{GD}, GPA_{3.6}\}$\\    
		& $t_4$            & M               & High     & $\{Activity_{RB}, GPA_{3.8}\}$\\
         & $t_4'$           & M               & High & $\{Activity_{TU}, GPA_{3.8}\}$\\
		 & $t_5$           & F               & Medium   & $\{Activity_{MO}, GPA_{3.6}\}$\\    
		& $t_6$            & F               & Low       & $\{Activity_{SO}, GPA_{3.7}\}$ \\    
		 & $t_7$           & M               & Low       & $\{Activity_{RB}, GPA_{3.7}\}$ \\    
		 & $t_8$           & F               & High      & $\{Activity_{RB}, GPA_{3.9}\}$\\
         & $t_{8}'$           & F               & High       & $\{Activity_{TU}, GPA_{3.9}\}$\\
		 & $t_{10}$           & M               & High  & $\{Activity_{RB}, GPA_{3.7}\}$\\    
		 & $t_{11}$           & F               & Low    & $\{Activity_{RB}, GPA_{3.8}\}$\\    
		 & $t_{12}$           & M               & Medium    & $\{Activity_{RB}, GPA_{4.0}\}$\\    
  	 &  $t_{14}$            & F               & Low       & $\{Activity_{RB}, GPA_{3.7}\}$\\ \cline{2-5} 
		\end{tabular}
  \label{tab:joined}
\end{table}

\paragraph*{Variables} Given a query $Q$ and a dataset $D$, for each categorical predicate in $\cat(Q)$ over an attribute $A$, we define a variable $A_{v} \in \{0, 1\}$ for each value $v$ in the domain of $A$ in the $D$. Intuitively, a solution to the MILP where $A_{v} = 1$ corresponds to a refinement that includes $A = v$ in the categorical predicates. For each numerical predicate $A \diamond C \in \num(Q)$, we define a variable $C_{A, \diamond}$ whose value is in the range of values of $A$ in $D$, and a set of variables $A_{v, \diamond}$ for each value $v$ in the domain of $A$ in $D$.

\begin{example}
$Activity_{RB}$ and $Activity_{SO}$ are two of the variables generated by the categorical predicate {\tt Activity = `RB'}
since these values are present in the database. 
The variable $C_{GPA, \geq}$ is generated from the numerical predicate {\tt GPA >= 3.7}. Additionally, the variable $GPA_{3.7, \geq}$ is generated since there exists a tuple in the data with the value $3.7$ in the {\tt GPA} attribute. 
\end{example} 

The value of $C_{A, \diamond}$ represents the value of the constant $C$ in the refinement query, and the variables $A_{v, \diamond}$ are used to determine whether a given tuple $t$ in $D$ (with the value $v$ in $A$) satisfies that predicate over $A$ in the refined query. More concretely, the variable $A_{v, \diamond}$ is used to reflect whether  $v \diamond C_{A, \diamond}$. Finally, we use a variable $r_t$ to denote the existence of a tuple $t$ in the output of a refinement query and a variable $s_t$ to indicate the position of $t$ in the output.

\paragraph*{Expressions}  We formulate a set of expressions such that the assignment generated by a solver to the MILP instance corresponds to the set of tuples selected by the corresponding refinement query. A tuple is part of a query's output if it satisfies its predicates set. We first define expressions for numerical predicates. Intuitively, a tuple $t$ with value $v$ in attribute $A$ satisfies the predicate $A \diamond C_{A, \diamond}$ if $v \diamond C_{A, \diamond}$. For lower-bound predicates, i.e., when $A \geq C$ or $A > C$, we model this using the following MILP expressions for each predicate in $\num(Q)$ and each value $V$ in the domain of $A$ in $D$.%
\begin{align}
\begin{split}
C_{A, \diamond} + M_A \cdot A_{v, \diamond} &\geq v + (1 - {\sf St}(\diamond)) \cdot \delta \label{eq:value_bounds_inline}\\
C_{A, \diamond} - M_A \cdot (1 - A_{v, \diamond}) &\leq v - {\sf St}(\diamond) \cdot \delta
\end{split}
\end{align} 
where $M_A$ is a constant larger than the maximum absolute value in the domain of the attribute $A$ in $D$, ${\sf St}(\diamond)$ is $1$ if $\diamond$ is a strict inequality and $0$ otherwise, and $\delta$ is some small constant added when $\diamond$ is strict in order to relax the inequality as MILP expressions do not support strict inequalities. We choose $\delta$ to be smaller than the smallest pairwise difference between the values in the domain of $A$, ensuring the relaxation does not include another value from the domain. For upper-bound predicates, we instead use the following set of expressions.
\begin{align}
\begin{split}
C_{A, \diamond} - M_A \cdot A_{v, \diamond} &\leq v - (1 - {\sf St}(\diamond)) \cdot \delta \label{eq:value_bounds_inline2}\\
C_{A, \diamond} + M_A \cdot (1 - A_{v, \diamond}) &\geq v + {\sf St}(\diamond) \cdot \delta
\end{split}
\end{align}
 Intuitively, the first expressions in (\ref{eq:value_bounds_inline}) and (\ref{eq:value_bounds_inline2}) ensure that $A_{v, \diamond}$ is $1$ if $v \diamond C_{A, \diamond}$ is true and the second expressions are used to ensure that $A_{v, \diamond}$ is 0 if $v \diamond C_{A, \diamond}$ is false. Together, they enforce that $A_{v, \diamond}$ is $1$ if and only if $v \diamond C_{A, \diamond}$ is true.

\begin{example}
\label{ex:numeric_constraint}
Continuing with our example, the following expressions are generated using the variables $C_{GPA, \geq}$ and $GPA_{3.7, \geq}$ for the numerical predicate {\tt GPA $\geq$ 3.7} in the \running{}.
\begin{align*}
    C_{GPA, \geq} + 5 \cdot GPA_{3.7, \geq} \geq 3.701 \\ 
    C_{GPA, \geq} - 5 \cdot (1 - GPA_{3.7, \geq}) \leq 3.7
\end{align*}
Here $M_A$ is set to $5$, a value greater than any value of the attribute {\tt GPA} in the data, and ${\sf St}(\diamond)$ is $0$ (since the inequality in the predicate is not strict). Consider an assignment that assigns the $3.7$ to $C_{GPA, \geq}$. This assignment corresponds to a query with the predicate {\tt GPA $\geq$ 3.7}. Assigning this value to the above expression results in 
\begin{align*}
    3.7 + 5 \cdot GPA_{3.7, \geq} \geq 3.701 \\ 
    3.7 - 5 \cdot (1 - GPA_{3.7, \geq}) \leq 3.7
\end{align*}
In this case, the value of $GPA_{3.7, \geq}$ should be $1$ as well, indicating that tuples with a value $\geq 3.7$ in the {\tt GPA} attribute meet the condition. Indeed, these expressions can be satisfied if and only if the variable $GPA_{3.7, \geq}$ is assigned the value $1$. Notice that adding $\delta$ in the first expression is necessary in order to guarantee that the only valid assignment for $GPA_{3.7, \geq}$ is $1$.

\end{example}

Next, we construct expressions that model the existence of a tuple in the query's output (represented using the variable $r_t$). The expressions should be able to model any possible refinement. Note that the output of a refinement may include tuples that are not part of the output of the original query. To this end, we use $\widetilde{Q}$ to denote the query obtained from $Q$ by omitting the selection predicates and any {\tt DISTINCT} statement. Intuitively, the output of $\widetilde{Q}$ over $D$ contains the output of every possible refinement query. A tuple $t$ is in the output of a query $Q$ if $t$ satisfies all the predicates in $Q$. 
To indicate whether $t$ is part of the output, we leverage the notion of {\it lineage}. The lineage of a tuple $t\in\widetilde{Q}(D)$ is the set of variables $A_v$ and $A_{v, \diamond}$ that correspond to the values of $t$ for each attribute in $\attr(Q)$: $\prov(t)=\{A_{t.A} \mid \forall \left( \bigvee_{c \in C} A = c \right) \in \cat(Q)\}~ \cup~ \{A_{t.A, \diamond} \mid \forall (A \diamond C) \in \num(Q)\}$. Table~\ref{tab:joined} shows the result of $\widetilde{Q}(D)$ in our running example with the lineage annotation for each tuple.

Since the value of each variable $A_{t.A} = 1$ or $A_{t.A, \diamond} = 1$ indicates the satisfaction of a predicate over $A$ by $t$, a tuple $t$ is in the output of $Q$ if all predicates in $Q$ are true for $t$, i.e., $\sum_{p \in \prov(t)} p = |\{\conds(Q)\}|$. We use this property to construct an expression that models the behavior of $r_t$ for each tuple $t\in\widetilde{Q}(D)$.
Note that tuples appearing once in $Q(D)$ may appear multiple times in  $\widetilde{Q}(D)$. E.g., when using {\tt DISTINCT} selection after a join operation, as the case in the \running{}, where the tuples denoted by $t_4$ and $t'_4$ represent the same student ID (that appears once in the output). To address this case, we define $S(t) = \{ t' \mid t' \in \widetilde{Q}(D),\ \forall a \in {\sf Distinct}(Q)\ t.a = t'.a,\ \widetilde{Q}(D)(t') < \widetilde{Q}(D)(t) \}$ where ${\sf Distinct}(Q)$ is the set of attributes selected distinctly by $Q$. Namely, for a tuple $t$, $S(t)$ is the set of tuples with the same values on attributes selected distinctly that are ranked higher than $t$. For instance, in our example $S(t_4) = t'_4$. Intuitively, at most one of $S(t)\cup\{t\}$ can appear in the output of the refined query (depending on its selection predicates). We therefore add the following expression to the MILP.

\begin{align}
\begin{split}
0 &\leq \sum_{\mathclap{p \in \prov(t)}} p + \sum_{t' \in S(t)} (1 - r_{t'}) - (|\conds(Q)| + |S(t)|) \cdot r_t \label{eq:tuple_in_ranking_inline} \\
  &\leq |\conds(Q)| + |S(t)| - 1
\end{split}
\end{align}
The lower bound of this expression prevents $r_t$ from being assigned $1$ if {\it not} all attributes of the tuple satisfy the predicate of the corresponding refinement or {\it any} tuples sharing its distinct values ranked better than it were already selected. Similarly, the upper bound is used to ensure that $r_t$ is assigned the value $1$ in case that {\it all} of the attributes of the tuple satisfy the predicate of the corresponding refinement and {\it none} of the tuples sharing its distinct values ranked better than it were already selected\footnote{Allowing the same entity (e.g., student ID in our example) appear multiple times in the output can be done by removing {\tt DISTINCT} from the input query.}.

\begin{example}
\label{ex:tuple-in-ranking}
Consider $t_6$ in Table~\ref{tab:joined}.
The lineage of $t_6$ is the set of variables $\{Activity_{SO}, GPA_{3.7, \geq}\}$ and $|\conds(Q)| = 2$, and the set $S(t_6)$ is empty given there is only $1$ tuple with ID $6$ in $\widetilde{Q}(D)$. Thus, the MILP instance has the expression 
$$0 \leq Activity_{SO} + GPA_{3.7, \geq} - 2 \cdot r_{t_6} \leq 1$$
Assuming $GPA_{3.7, \geq} = 1$ and $Activity_{SO} =1$, $r_{t_6}$ must be assigned $1$, indicating that $t_6$ is part of the output in this case.

\end{example}

Given these $r_t$ variables, we may enforce that there at least $k^*$  tuples in the output of the refinement by the expression
\begin{align}
    \sum_{\mathclap{t \in \widetilde{Q}(D)}} r_{t} \geq k^* \label{eq:at_least_k_star_inline}
\end{align}

The last part required to complete the correspondence between the solution to the MILP instance and the output of a refinement query is modeling the order of the output tuples (according to the {\tt ORDER BY} expression of the input query) through the MILP expressions. We use the set of variables $s_t$ for each tuple in $\widetilde{Q}(D)$, which represents the position of $t$ in the output of the corresponding refinement query. Intuitively, the position of a tuple $t$ in the output of the refinement query $Q'$ is one plus the number of tuples $t'\in \widetilde{Q}(D)$ that are part of the output $Q'$ and ranked higher than $t$ (i.e., $\widetilde{Q}(D)(t') < \widetilde{Q}(D)(t)$). For tuples $t$ that are not part of the output of $Q'$, the variable $s_t$ will be assigned a value larger than $|\widetilde{Q}(D)|$. This is modeled using the following set of expressions.
\begin{align}
    1 + |\widetilde{Q}(D)| \cdot (1 - r_t) + \sum_{\mathclap{\substack{t' \in \widetilde{Q}(D),\\ \widetilde{Q}(D)(t') < \widetilde{Q}(D)(t)}}} r_{t'} = s_{t} \label{eq:position_in_ranking_inline}
\end{align}
for each $t$ in $\widetilde{Q}(D)$. Given this, we may further limit $s_t$ to be in the range $[1, 2 \cdot |\widetilde{Q}(D)|]$.

\begin{example}
\label{ex:position}
In our running example $|\widetilde{Q}(D)| = 14$. Thus the expression $1 + 14 \cdot (1 - r_{t_6}) + r_{t_1} + r_{t_2} + r_{t_3} + r_{t_4} + r_{t_4'} + r_{t_5} = s_{t_6}$ is  the expressions generated for the tuple $t_6$. Assuming $r_{t_1}$, $r_{t_2}$, $r_{t_4}$ and $r_{t_6}$ are $1$ (and the rest of the variables are $0$), the value of $s_{t_6}$ must be $4$, indicating its position in the ranking in this case.

\end{example}

\subsection{Bounding Maximum Deviation}

The second part of the solution consists of expressions whose goal is to limit the refinement query's output's deviation from the constraint set $\constraints$ to be at most $\varepsilon$. For each cardinality constraint $\mathscr{c}_{G, k} = n$ in $\constraints$, we are interested in the number of tuples belonging to group $G$ in the top-$k$ of the refined ranking to determine the number of tuples of group $G$ needs to be added or removed to satisfy $\mathscr{c}_{G, k} = n$.
To model this property, we introduce two sets of new variables $l_{t, k}$ and $E_{G, k}$. The variables $l_{t, k}$ are used to indicate whether a tuple $t$ appears in the top-$k$ ranked output of the corresponding refinement query, and $E_{G, k}$ represents the number of tuples from $G$ in the top-$k$ that need to be added (removed) to satisfy lower-bound (upper-bound) cardinality constraints (i.e., $E_{G, k}$ is equivalent to the numerator in the summation of \Cref{def:mospe} for each cardinality constraint). Intuitively, we may further specify that $E_{G, k} \in [0, k]$. 

We use a similar construction to the expressions in (\ref{eq:value_bounds_inline}) to ensure that $l_{t, k} = 1$ if and only if the tuple $t$ appears in the top-$k$ as follows.
\begin{align}
\begin{split}
    s_t + (2 \cdot |\widetilde{Q}(D)| + 1) \cdot l_{t, k} &\geq k + \delta \label{eq:in_prefix_inline} \\
    s_t - (2 \cdot |\widetilde{Q}(D)| + 1) \cdot (1 - l_{t, k}) &\leq k
\end{split}
\end{align}
where $(2 \cdot |\widetilde{Q}(D)| + 1)$ is the constant coefficient that plays the rule of $M_A$ in (\ref{eq:value_bounds_inline}) and $\delta$ is a small additive constant as in (\ref{eq:value_bounds_inline}) and (\ref{eq:value_bounds_inline2}). 

\begin{lemma}
\label{thm:same-top-k}
    Let $D$ be a dataset, $Q$ be a query over $D$, and $\constraints$ be a set of cardinality constraints. There is an assignment $\alpha$ satisfying the expressions generated by (\ref{eq:value_bounds_inline}-\ref{eq:in_prefix_inline}) if and only if there is a refinement $Q'$ of Q such that
    \begin{enumerate}[label=\numbercircled{\arabic*}]
        \item For each $(\bigvee_{c \in C} A = c) \in \cat(Q')$, $\alpha(A_c) = 1 \iff c \in C$
        \item For each $(A \diamond C) \in \num(Q')$, $C_{A, \diamond} = C$
        \item For each $(\mathscr{c}_{G, k} = n) \in \constraints$, $t \in Q'(D)_k \cap G \iff \alpha(l_{t, k}) = 1$
    \end{enumerate}
\end{lemma}
\begin{proof}
    \underline{Assignment $\Longrightarrow$ refinement:} We may directly obtain from this assignment a refinement that satisfies properties \numbercircled{1} and \numbercircled{2}. Now, we only have left to show that this refinement $Q'$ has property \numbercircled{3}. By expression (\ref{eq:tuple_in_ranking_inline}), we have that for a tuple $t \in \widetilde{Q}(D)$ that $\alpha(r_t) = 1$ if and only if $t$ matches the predicates of $Q'$ and no tuple ranking better than $t$ sharing its same distinct values is selected by $Q'$ (i.e. $t \in Q'(D)$). Recall that by expression (\ref{eq:position_in_ranking_inline}), for any tuple $t \in Q'(D)$ we have that $\alpha(s_t) = Q'(D)(t)$ as $\alpha(r_t) = 1$. Finally, by expression (\ref{eq:in_prefix_inline}), we have for any tuple $t \in Q'(D)$ that $\alpha(l_{t, k}) = 1$ if and only if $\alpha(s_t) \leq k$. This gives that for any tuple $t$ in $\widetilde{Q}(D)$ that $t \in Q'(D)_k \iff \alpha(l_{t, k}) = 1$.  Given that for any $(\mathscr{c}_{G, k} = n) \in \constraints$, $Q'(D)_k \cap G \subseteq Q'(D)_k$, we have that $Q'$ has property \numbercircled{3} as desired. This shows that an assignment satisfying  the expressions generated by (\ref{eq:value_bounds_inline}-\ref{eq:in_prefix_inline}) produces a refinement with the stated properties. \\
    \indent \underline{Refinement $\Longrightarrow$ assignment:} For each $(A \diamond C) \in \num(Q')$ and each value $v$ in the domain of $A$, we make the assignment $\alpha(A_{v, \diamond}) = 1$ if and only if $v \diamond C$, and assign $\alpha(A_{v, \diamond}) = 0$ otherwise. Such assignments make the expressions generated by (\ref{eq:value_bounds_inline}) and (\ref{eq:value_bounds_inline2}) feasible. Then, for a tuple $t \in \widetilde{Q}(D)$, we make an assignment such that $\alpha(r_t) = 1 \iff t \in Q'(D)$, and $\alpha(r_t) = 0$ otherwise. These assignments make the expressions generated by (\ref{eq:tuple_in_ranking_inline}) feasible as for $t \in Q'(D)$, we have that $t$ matches the predicates of $Q'$ and none of the tuples ranking better than $t$ sharing its distinct values are selected by $Q'$ as well, making $\alpha(r_t) = 1$ the only feasible assignment for such a case. Otherwise, we have a tuple that is in $\widetilde{Q}(D)$ but not $Q'(D)$, so it either did not match at least one of the predicates of $Q'$ or at least one of the tuples ranking better than it sharing its distinct values was selected by $Q'$, therefore $\alpha(r_t) = 0$ is the only feasible assignment. We now make the assignments $\alpha(s_t) = Q'(D)(t)$ for each tuple $t \in Q'(D)$, and $\alpha(s_t) = 1 + |\widetilde{Q}(D)| + \sum_{t' \in \widetilde{Q}(D),~\widetilde{Q}(D)(t') < \widetilde{Q}(D)(t)} r_{t'}$ for each tuple $t \in \widetilde{Q}(D) \setminus Q'(D)$. These assignments are feasible for expression (\ref{eq:position_in_ranking_inline}) given that for a tuple in $Q'(D)$, its position $Q'(D)(t)$ is equal to $1 + \sum_{t' \in \widetilde{Q}(D),~\widetilde{Q}(D)(t') < \widetilde{Q}(D)(t)} r_{t'}$. For each tuple $t \in \widetilde{Q}(D) \setminus Q'(D)$, the feasibility of the assignment $\alpha(s_t)$ is immediate. Finally, for each tuple $t \in \widetilde{Q}(D)$ and $(\mathscr{c}_{G, k} = n) \in \constraints$, we make the assignment $\alpha(l_{t, k}) = 1$ if and only if $t \in Q'(D)_k$, and $\alpha(l_{t, k}) = 0$ in any other case. While we assume some of these assignments from property \numbercircled{3}, we must make these assignments to ensure that the expressions generated by \ref{eq:in_prefix_inline} are feasible for tuples in $Q'(D)_k$, but not in any group with a constraint in $\constraints$. This shows that we may construct an assignment satisfying the expressions generated by (\ref{eq:value_bounds_inline}-\ref{eq:in_prefix_inline}) from a refinement with the stated properties.
\end{proof}

We utilize the variables $l_{t, k}$ to determine the values of $E_{G, k}$ using the following expressions for each cardinality constraint $\mathscr{c}_{G, k} = n$~in~$\constraints$.
\begin{align}
\begin{split}
E_{G, k} &\geq 0 \label{eq:tuples_to_satisfy_inline} \\
E_{G, k} &\geq \sign(\mathscr{c}) \cdot \left (n - \sum_{t \in \widetilde{Q}(D)\cap G} l_{t, k} \right)
\end{split}
\end{align}
where $\sum_{t \in \widetilde{Q}(D)\cap G} l_{t, k}$ is the number of tuples belonging to group $G$ in the top-$k$.

Finally, to restrict the deviation of the refinement's output to at most $\varepsilon$, we construct the following expression
\begin{align}
\frac{1}{|\constraints|}\sum_{(\mathscr{c}_{G, k} = n) \in \constraints} \frac{E_{G, k}}{n} \leq \varepsilon \label{eq:max_deviation_inline}
\end{align}

\begin{example}
\label{ex:deviation}
Consider again the database shown in \Cref{ex:running} and the cardinality constraint $\lb{Gender=`Female'}{k=6} = 3$. %
Tuples $t_2$, $t_3$, $t_5$, $t_6$, $t_8$, $t_8'$, $t_{11}$, and $t_{14}$ are in the group  \{Gender=`Female'\}. Thus, we generate the expressions
\begin{align*}
\begin{split}
E_{Gender=`Female', 6} \geq&~ 0 \\
E_{Gender=`Female', 6} \geq& ~3 - (l_{t_2, 6}+l_{t_3, 6}+l_{t_5, 6}+l_{t_6, 6}\\
&~+l_{t_8, 6}+l_{t_8', 6}+l_{t_{11}, 6}+l_{t_{14}, 6})
\end{split}
\end{align*}
where the value of $l_{t_6, 6}$, for instance, is used in the expressions
\begin{align*}
    s_{t_6} + 25 \cdot l_{t_6, 6} &\geq 6.001 \\
    s_{t_6} - 25 \cdot (1 - l_{t_6, 6}) &\leq 6
\end{align*}
Continuing Example~\ref{ex:position}, assuming $s_{t_6}$ is assigned the value $4$, forcing the assignment of $1$ to $l_{t_6, 6}$. Assuming $s_{t_2} = 2$ and $s_{t_8} = 6$, we would similarly get that the value of $l_{t_2, 6}$ and $l_{t_8, 6}$ must be $1$. Using these values in expression generated by~(\ref{eq:tuples_to_satisfy_inline}) results in 
\begin{align*}
E_{Gender=`Female', 6} &\geq 0 \\
E_{Gender=`Female', 6} &\geq 3 - (1+1+1) = 0
\end{align*}
This intuitively means that no additional tuples from the group {\tt {Gender=`Female'}} are required to satisfy the constraint.

\end{example}

We summarize our mixed-integer linear program in \Cref{fig:milp}, and its variables in \Cref{tab:variables}. %
By satisfying all of these expressions together, we produce rankings that are both valid and sufficiently satisfactory of the constraint set. In fact, we can show that any satisfying assignment $\alpha$ to the variables in the expressions generated by (\ref{eq:value_bounds_inline})-(\ref{eq:max_deviation_inline}) corresponds to a valid refinement that is sufficiently satisfactory.
\begin{theorem}[Solution correctness]
\label{thm:bounded-deviation}
    Let $D$ be a dataset, $Q$ a query over $D$, $\constraints$ be a set of cardinality constraints, and $\varepsilon$ a threshold over the deviation from $\constraints$. There is an assignment $\alpha$ satisfying the expressions generated by (\ref{eq:value_bounds_inline}-\ref{eq:max_deviation_inline}) if and only if there is a refinement $Q'$ for $Q$ such that 
    \begin{enumerate}[label=\numbercircled{\arabic*}]
        \item For each $(\bigvee_{c \in C} A = c) \in \cat(Q')$, $\alpha(A_c) = 1 \iff c \in C$
        \item For each $(A \diamond C) \in \num(Q')$, $\alpha(C_{A, \diamond}) = C$
        \item $DEV(Q'(D), \constraints) \leq \frac{1}{\constraints}\sum_{(\mathscr{c}_{G, k} = n) \in \constraints} \frac{\alpha(E_{G, k})}{n} \leq \varepsilon$
    \end{enumerate}
\end{theorem}

\begin{proof}
    \underline{Assignment $\Longrightarrow$ refinement:} By \Cref{thm:same-top-k}, there is a refinement with properties \numbercircled{1} and \numbercircled{2} as expressions (\ref{eq:value_bounds_inline}-\ref{eq:in_prefix_inline}) are satisfied by assumption. This refinement also has the property that for each $(\mathscr{c}_{G, k} = n) \in \constraints$, $t \in Q'(D)_k \cap G \iff \alpha(l_{t, k}) = 1$. Therefore, we have as a corollary that for each $(\mathscr{c}_{G, k} = n) \in \constraints$, $|Q'(D)_k \cap G| = \sum_{t \in \widetilde{Q}(D) \cap G}\alpha(l_{t, k})$. Given this fact, then by \Cref{def:mospe} and our assumption that the assignment satisfies expressions (\ref{eq:tuples_to_satisfy_inline}-\ref{eq:max_deviation_inline}), this refinement also has property \numbercircled{3}. Therefore, there is a refinement with properties \numbercircled{1}-\numbercircled{3} given an assignment satisfying expressions (\ref{eq:value_bounds_inline}-\ref{eq:max_deviation_inline}).\\
    \indent \underline{Refinement $\Longrightarrow$ assignment:} By assigning $\alpha(l_{t, k})$ to each $l_{t, k}$ variable such that for each $(\mathscr{c}_{G, k} = n) \in \constraints$, $\alpha(l_{t, k}) = 1$ iff $t \in Q'(D)_k$ (further noting that $Q'(D)_k \cap G \subseteq Q'(D)$), and assuming properties \numbercircled{1} and \numbercircled{2}, we have by \Cref{thm:same-top-k} that there is an assignment satisfying expressions (\ref{eq:value_bounds_inline}-\ref{eq:in_prefix_inline}). Expression (\ref{eq:max_deviation_inline}) is satisfied by assumption, and expression (\ref{eq:tuples_to_satisfy_inline}) is therefore satisfied given \Cref{def:mospe} and that for $(\mathscr{c}_{G, k} = n) \in \constraints$ we have $|Q'(D)_k \cap G| = \sum_{t \in \widetilde{Q}(D) \cap G}\alpha(l_{t, k})$. Therefore, there is an assignment satisfying expressions (\ref{eq:value_bounds_inline}-\ref{eq:max_deviation_inline}) given a refinement with properties \numbercircled{1}-\numbercircled{3}. 
\end{proof}

\begin{figure}
    \centering
    \footnotesize
    \begin{align*}
        \min \quad & DISTANCE && \notag \\
        \textrm{s.t. } \quad & C_{A, \diamond} + M_A \cdot A_{v, \diamond} \geq v + (1 - {\sf St}(\diamond)) \cdot \delta && \forall (A \diamond C) \in \num^{>}(Q) \\
        & C_{A, \diamond} - M_A \cdot (1 - A_{v, \diamond}) \leq v - {\sf St}(\diamond) \cdot \delta && \forall (A \diamond C) \in \num^{>}(Q) \\
        & C_{A, \diamond} - M_A \cdot A_{v, \diamond} \leq v - (1 - {\sf St}(\diamond)) \cdot \delta && \forall (A \diamond C) \in \num^{<}(Q) \\
        & C_{A, \diamond} + M_A \cdot (1 - A_{v, \diamond}) \geq v + {\sf St}(\diamond) \cdot \delta && \forall (A \diamond C) \in \num^{<}(Q) \\ 
        & \begin{aligned}
            0 &\leq \sum_{\mathclap{p \in \prov(t)}} p + \sum_{t' \in S(t)} (1 - r_{t'})\\
              &- (|\conds(Q)| + |S(t)|) \cdot r_t \\
              &\leq |\conds(Q)| + |S(t)| - 1
        \end{aligned}&& \forall t \in \widetilde{Q}(D)  \\
        & \sum_{\mathclap{t \in \widetilde{Q}(D)}} r_{t} \geq k^* \\
        & 1 + |\widetilde{Q}(D)| \cdot (1 - r_t) + \sum_{\mathclap{\substack{t' \in \widetilde{Q}(D),\\ \widetilde{Q}(D)(t') < \widetilde{Q}(D)(t)}}} r_{t'} = s_{t} && \forall t \in \widetilde{Q}(D) \\
        & s_t + (2 \cdot |\widetilde{Q}(D)| + 1) \cdot l_{t, k} \geq k + \delta && \forall t \in \widetilde{Q}(D), (\mathscr{c}_{G, k} = n) \in \constraints \\
        & s_t - (2 \cdot |\widetilde{Q}(D)| + 1) \cdot (1 - l_{t, k}) \leq k && \forall t \in \widetilde{Q}(D), (\mathscr{c}_{G, k} = n) \in \constraints \\
        & E_{G, k} \geq 0 && \forall (\mathscr{c}_{G, k} = n) \in \constraints \\
        & E_{G, k} \geq \sign(\mathscr{c}) \cdot \left (n - \sum_{t \in \sigma_G(\widetilde{Q}(D))} l_{t, k} \right ) && \forall (\mathscr{c}_{G, k} = n) \in \constraints \\
        & \frac{1}{|\constraints|} \sum_{(\mathscr{c}_{G, k} = n) \in \constraints} \frac{E_{G, k}}{n} \leq \varepsilon 
    \end{align*}
    \caption{Summary of our MILP model}
    \label{fig:milp}
\end{figure}

\begin{figure}[h]
    \centering
    \includegraphics[width=\linewidth]{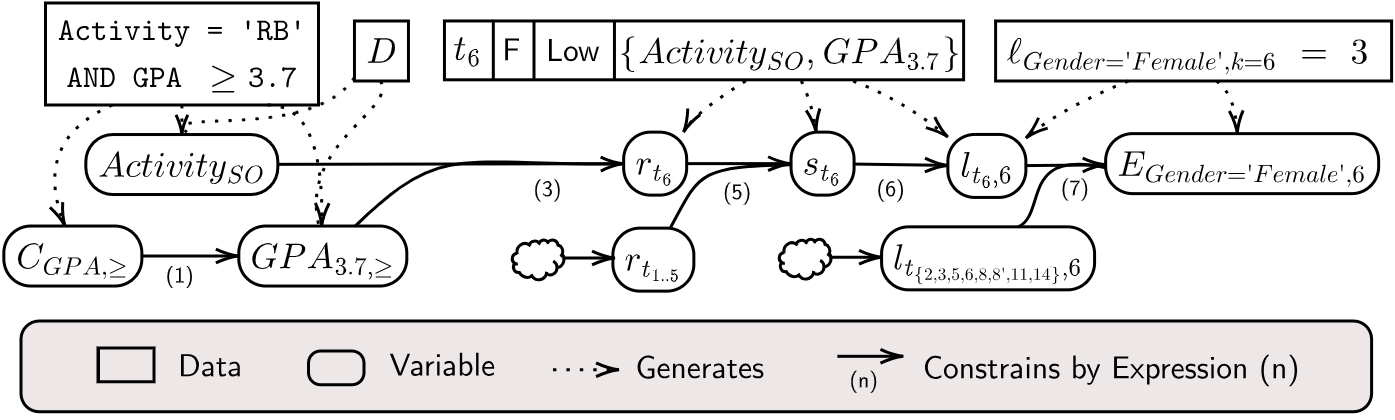}
    \caption{Diagram illustrating the expression generation for our running example. The predicate {\tt Activity = `RB' AND GPA $\geq$ 3.7} generates the variables $Activity_{SO}$ and $GPA_{3.7, \geq}$ as `SO' and $3.7$ are values that appear for those attributes respectively in the database $D$. $C_{GPA, \geq}$ is also generated by the predicate to hold the new constant of the predicate, and constrains the value of $GPA_{3.7, \geq}$ by (\ref{eq:value_bounds_inline}). The tuple $t_6\in \widetilde{Q}$ generates the variable $r_6$, whose value is constrained through (\ref{eq:tuple_in_ranking_inline}) by the values of $Activity_{SO}$ and $GPA_{3.7, \geq}$ due to its lineage. It also generates the variable $s_6$, which is then constrained by the value of the $r_t$ values for the tuples that rank better than it, i.e., $r_{t_{1..5}}$, through (\ref{eq:position_in_ranking_inline}). Finally, the constraint $\lb{Gender=`Female'}{k=6} = 3$ combines with $t_6$ to generate the variable $l_{t_{6},6}$ which is constrained by the value of $s_{t_6}$ by (\ref{eq:in_prefix_inline}). The constraint generates the variable $E_{Gender=`Female',6}$ which is constrained through (\ref{eq:tuples_to_satisfy_inline}) by the values of all the $l_{t, 6}$ variables for which $t$ is a part of the group (listed in \Cref{ex:deviation}).}
    \label{fig:example_summary}
    \vspace{-0.5cm}
\end{figure}

\paragraph*{Model limitation}
\label{sec:query_class_lims}
Given an input to our program, we construct a MILP program. The correctness of the solution generated by the program, as stated in \Cref{thm:bounded-deviation}, relies on three properties. First, every possible refinement may be represented as an assignment to the variables of the program. Second, the tuples in the output of any potential refinement are in the same relative order, and finally, every tuple in the output satisfies all the predicates of the corresponding refinement. We define the problem for SPJ queries, and thus, our model is designed to handle SPJ queries, and these properties hold for them. We note that supporting other classes of queries may require modifications to the problem definition as well as to our proposed model. For instance, in union queries, it is enough for a tuple in the output to satisfy the predicates of one branch of the union, in contrast to the third property. This may be handled straightforwardly, as noted in \Cref{sec:background}. Handling nested queries is more challenging since they may contain multiple selection statements at different nesting levels. The problem definition should first be extended to properly define how such a query can be refined, e.g., whether refinements at different nesting levels are allowed. Our proposed model cannot capture the refinement of selection statements in different nesting levels and, therefore, does not fulfill the first property. Moreover, if the {\tt ORDER BY} clause relates to an inner query, refining the inner query may change the relative order of the tuples in contrast to the second property.

\section{Optimizations}
\label{sec:optimizations}
In Section~\ref{sec:search}, we have presented a MILP formulation designed to solve the \problem{} problem. While this approach enables us to leverage existing MILP solvers, that can solve the problem efficiently, they often encounter difficulties when dealing with extensive programs (containing numerous expressions), and a large number of variables~\cite{QFix}. While the number of expressions and variables in the MILP we generate is linear in the data size, as we show in Section~\ref{sec:experiments}, MILP solvers struggle to scale and solve the generated programs.

To this end, we propose three optimizations for the construction of the MILP problem: one is a general optimization that applies in all cases, and the other two are limited in some cases. The first optimization is relevancy-based and removes from consideration tuples that are irrelevant to determining the satisfaction of the constraint set. The second optimization reduces the number of binary variables in the MILP problem by combining redundant variables. This optimization cannot be applied for queries with a {\tt DISTINCT} statement. %
The third optimization relaxes the expression used to determine the score of a tuple in the new ranking. This optimization is only applicable to tuples belonging to groups with only lower-bound or only upper-bound constraints, but not both.

\paragraph*{Relevancy-Based}
We propose a relevancy-based optimization to reduce the number of expressions and variables in our problem.
Recall that we use $k^*$ to denote the maximal $k$ that appears in the constraint set $\constraints$. Then, by removing tuples that could never appear in the top-$k^*$ in any refinement, we are able to avoid adding their variables and expressions to our problem. We determine the relevancy of these tuples by selecting the top-$k^*$ of the groups of tuples that share the same lineage. Let $[\prov(t)]$ be the equivalence class of tuples that share the same lineage as a tuple $t$. Then for a tuple $t$, let $T(t)$ be the ranking generated by ranking the tuples of $[\prov(t)]$ according to the {\tt ORDER BY} clause of $Q$. We see trivially that it is not possible for tuples past position $k^*$ in $T(t)$ for all $t$ in $\widetilde{Q}(D)$ to be included in the top-$k^*$ of any refinement. Thus, it is sufficient to consider only the top-$k^*$ of $T(t)$, denoted by $T(t)_{k^*}$, in the generated program, and we replace $\widetilde{Q}(D)$ in the expressions referencing it with $T(t)_{k^*}$ in Figure~\ref{fig:milp}.

\begin{example}\label{ex:optimization1}
    Consider $t_{14}$ from Table~\ref{tab:joined}. Its equivalence class $[\prov(t_{14})]$ is the set $\{t_7, t_{10}, t_{14}\}$. Assume we are interested in satisfying a single constraint $\lb{Gender=`Female'}{2} = 1$. Note that the tuple $t_{14}$ can never appear in the top-$2$ of any refinement query, as any refinement that includes $t_{14}$ includes tuples with its same lineage, i.e., $t_7$ and $t_{10}$. Therefore, it is safe to remove all variables and expressions related to $t_{14}$ from consideration.
\end{example}
This optimization is most effective when $k^*$ is small, and there are few lineage equivalence classes. In \Cref{sec:experiments}, we show that this is often the case in queries over real data sets. We further demonstrate the effect of $k^*$ on the running time in \Cref{fig:time_vs_k}.

\paragraph*{Selecting Lineages} 
Recall that the program we generate includes a binary variable $r_t$ for each tuple in $\widetilde{Q}(D)$. However, tuples sharing the same lineage all have equal values for their $r_t$ variables. Therefore, we can use a single variable for all tuples in the same lineage equivalence classes. 

\begin{example}
    To demonstrate this idea, consider the \running{} without its {\tt DISTINCT} statement and consider again the tuple $t_{14}$ with its equivalence class shown in Example~\ref{ex:optimization1}. If $t_{14}$ satisfies the selection condition of a refinement on the \running{}, then $t_7$ and $t_{10}$ must satisfy the conditions as well, as they share the same lineage. Therefore, we have the equivalence $r_{t_{14}} = r_{t_7} = r_{t_{10}}$. The variables $r_{t_7}$ and $r_{t_{10}}$ are then made redundant, as they always have the same value as $r_{t_{14}}$. 
\end{example}

In order to avoid such redundancy, instead of constructing the set of $r_t$ variables, we construct a set of variables $r_{[\prov(t)]}$ for every tuple $t$ in $\widetilde{Q}(D)$. Using (\ref{eq:tuple_in_ranking_inline}) as a basis, we are able to model $r_{[\prov(t)]}$ being assigned $1$ if and only if the tuples in $[\prov(t)]$ satisfy the selection condition of the corresponding refinement query.
 Instead of constructing expression (\ref{eq:tuple_in_ranking_inline}) for each tuple in $\widetilde{Q}(D)$, we construct the following expression for each $r_{[\prov(t)]}$ variable:
    $0 \leq \sum_{p \in \prov(t)} p - |\conds(Q)| \cdot r_{[\prov(t)]} \leq |\conds(Q)| - 1$.
Furthermore, in order to ensure that the $s_t$ values are modeled as before, we modify (\ref{eq:position_in_ranking_inline}) by changing $r_t$ to $r_{[\prov(t)]}$ and $r_{t'}$ to $r_{[\prov(t')]}$.
We note that this optimization cannot be applied if the input query includes a {\tt DISTINCT} statement,
as we need this information in order to not select tuples that already have a tuple sharing its distinct value selected.

\begin{table*}[t!]
    \centering
    \footnotesize  
    \begin{tabular}{lclcl}
    \hline
    \textbf{Dataset} & \textbf{Query} & \multicolumn{1}{c}{\textbf{Predicates}} & \textbf{Order by (DESC)} & \multicolumn{1}{c}{\textbf{Constraints}}\\ \hline\hline
    Astronauts & $Q_A$ & 
    \begin{tabular}{l@{}}\texttt{"Graduate Major" = 'Physics'}\\ \texttt{AND "Space Walks" <= 3}\\ \texttt{AND "Space Walks" >= 1} \end{tabular}  
    &{\tt "Space Flight (hrs)"}&\begin{tabular}{@{}ll@{}}
        (1) $\lb{Gender='F'}{k} = \frac{k}{2}$ & (4) $\lb{Status='Management'}{k} = \frac{k}{5}$ \\ (2) $\lb{Gender='M'}{k} = \frac{k}{2}$ & (5) $\lb{Status='Retired'}{k} = \frac{k}{5}$ \\ (3) $\lb{Status='Active'}{k} = \frac{k}{5}$ &
    \end{tabular}\\
    \hline
    Law Students & $Q_L$ & \begin{tabular}{l@{}}\texttt{Region = 'GL'}\\\texttt{AND GPA <= 4.0}\\\texttt{AND GPA >= 3.5}\end{tabular}& {\tt LSAT}&\begin{tabular}{@{}ll@{}}
        (1) $\lb{Sex='F'}{k} = \frac{k}{2}$ & (4) $\lb{Race='White'}{k} = \frac{k}{5}$ \\ (2) $\lb{Sex='M'}{k} = \frac{k}{2}$ & (5) $\lb{Race='Asian'}{k} = \frac{k}{5}$ \\ (3) $\lb{Race='Black'}{k} = \frac{k}{5}$ &
    \end{tabular}\\
   \hline
    MEPS & $Q_M$ & \begin{tabular}{l@{}}\texttt{Age > 22}\\\texttt{AND "Family Size" >= 4}\end{tabular}& {\tt Utilization}&\begin{tabular}{@{}ll@{}}
        (1) $\lb{Sex='F'}{k} = \frac{k}{2}$ & (4) $\lb{Race='Black'}{k} = \frac{k}{5}$ \\ (2) $\lb{Sex='M'}{k} = \frac{k}{2}$ & (5) $\lb{Race='White'}{k} = \frac{k}{5}$ \\ (3) $\lb{Race='Asian'}{k} = \frac{k}{5}$ &
    \end{tabular}\\ \hline
    TPC-H & $Q_5$ & \texttt{Region = 'ASIA'} & {\tt Revenue}& \begin{tabular}{@{}ll@{}}
        (1) $\lb{OrderPrio='5-LOW'}{k} = \frac{k}{2}$ & (4) $\lb{MktSeg='BUILDING'}{k} = \frac{k}{5}$ \\ (2) $\lb{OrderPrio='3-MEDIUM'}{k} = \frac{k}{5}$ & (5) $\lb{MktSeg='MACHINERY'}{k} = \frac{k}{5}$ \\ (3) $\lb{MktSeg='AUTOMOBILE'}{k} = \frac{k}{5}$ &
    \end{tabular}\\ \hline
    \end{tabular}
    \caption{Queries and constraints}
    \label{tab:queriesAndConstraints}
    \vspace{-0.5cm}
\end{table*}

\paragraph*{Relaxation for Single-Constraint-Type Tuples}
We present another optimization that is possible when a tuple belongs to groups that have only either lower-bound ($\ell$) or upper-bound ($\mathscr{u}$) cardinality constraints made on them. We define the set of tuples belonging only to groups with lower-bound constraints as 
$L = \{ t \mid \not\exists (\ub{G}{k} = n) \in \constraints, t \in \widetilde{Q}(D) \cap G \}$.
We define a similar set $U$ for upper-bound tuples, replacing $\ub{G}{k} = n$ in the quantifier with $\lb{G}{k} = n$. Then, for a tuple $t \in L$, we relax expression (\ref{eq:position_in_ranking_inline}) to 
$1 + |\widetilde{Q}(D)| \cdot (1 - r_t) + \sum_{t' \in \widetilde{Q}(D), \widetilde{Q}(D)(t') < \widetilde{Q}(D)(t)} r_{t'} \leq s_{t}$.

For tuples in $U$ we use an upper bound instead ($\geq s_t$). This relaxation makes finding feasible solutions for this model easier and may be used by presolving techniques in MILP solvers.
To understand why this maintains the correctness of our solution, consider the lower-bound constraints. Intuitively, we can allow the $s_t$ variables of tuples belonging to the group defined in the constraint in a given top-$k$ to be assigned a value larger than the position of $t$ in the ranking, as this can only result in a worse deviation from satisfying the constraint (but can not assign a value smaller than the position of $t$ in the ranking). The case for upper-bound constraints is symmetric.

\section{Experiments}
\label{sec:experiments}

We performed an experimental analysis of our proposed algorithm on real-life and synthetic datasets considering realistic scenarios. We first examine the effect of different parameters on the running time. 
We show that our solution scales, performs well on realistic scenarios and that the optimization presented in Section~\ref{sec:optimizations} are effective. We then compare our solution to~\cite{MLJ22,ERICA} that studies a similar problem for queries without ranking. We demonstrate the differences between solutions and compare their outputs and performance through a use case.

\subsection{Evaluation Benchmark}\label{sec:benchmark}
To the best of our knowledge, we are the first to consider this problem, and there is no benchmark consisting of datasets, including ranking queries and sets of cardinality constraints. To this end, we have developed a dedicated benchmark that involves real-life datasets used in the context of ranking as follows. 
\begin{itemize}[leftmargin=1em,labelwidth=*,align=left]
    \item {\bf Astronauts\footnote{\url{https://www.kaggle.com/datasets/nasa/astronaut-yearbook}}}: A dataset of 19 attributes containing 357 NASA astronauts and information about their careers. Astronauts are ranked in descending order by their number of space flight hours, as was done in \cite{SYJ18}.
    \item {\bf Law Students \cite{LawDataOriginal,LawData}}: A dataset of 8 attributes containing 21,790 law students and various evaluations such as grade point average, LSAT examination scores, and first year grade average. Students are ranked by their LSAT scores, as in \cite{ZHW20}.
    \item {\bf MEPS\footnote{\url{https://meps.ahrq.gov/data\_stats/download\_data/pufs/h192/h192doc.shtml}}}: A dataset of 1,941 attributes containing 34,655 individuals and information related to their usage of healthcare. Patients are ranked in descending order by a combination of utilization metrics (office-based visits + ER visits + in-patient nights + home health visits), as was done in \cite{YGS19}. 
\end{itemize}

To evaluate scalability, we use Synthetic Data Vault (SDV) \cite{SDV} to learn the distributions of our real-life datasets and subsequently synthesize scaled-up versions.
We also use the {\bf TPC-H Benchmark}%
, which includes complex queries involving multiple tables. We generate a TPC-H dataset of scale factor $1$, which is approximately $1$ GB of data. We use Query 5 (Q5) from the TPC-H specification and remove the predicates filtering on date types.

\subsubsection*{Queries and constraints} %
\Cref{tab:queriesAndConstraints} summarizes the queries and constraints used.
 We generated queries and constraints for each dataset, showcasing real-life scenarios. Each row in the table represents a query. %
 For example, the first line represents the following query 
 $Q_A$ over the Astronauts dataset. 
 \begin{center}
     \footnotesize
     \begin{tabular}{l}
        \verb"SELECT * FROM Astronauts" \\
        \verb|WHERE "Space Walks" <= 3 AND "Space Walks" >= 1| \\
        \verb|AND "Graduate Major" = 'Physics'|\\
        \verb|ORDER BY "Space Flight (hrs)" DESC|\\
     \end{tabular}
     \end{center}
     This query may be used in the selection process of astronauts for a mission. The mission requires specific training (number of space walks) and background (graduate major), and the candidates are ordered by their experience (space flight hours).  
     Similarly, the query $Q_L$ for the Law Students dataset may be used to rank outstanding students (based on their GPA) from a particular region based on their SAT scores for a scholarship. Finally, $Q_M$ is defined for the MEPS dataset.  Such a query may be used to invite the best-fitting patients (based on their utilization) with specific criteria, for a study.

     We defined result diversity constraints for each dataset (listed in Table~\ref{tab:queriesAndConstraints}). For instance, in the Astronauts dataset, the result should include women and candidates of varying ranks in the organizational hierarchy. The constraints' bounds are parameterized with a value $k$, and we set them to values that produce a valid refinement in most cases. Specifically, out of 132 performed experiments, we were not able to find a solution in only 2.

\subsubsection*{Parameters setting} When using ranked-retrieval in decision-making contexts (e.g., when deciding how many people to invite for in-person interviews), one expects the number of items a user will consider ($k$) to be relatively low. In general, rankings are subject to position bias --- a geometric drop in visibility of items in lower ranks --- and so are best-suited for cases where the user interacts with a small number of top-ranked items~\cite{DBLP:journals/cacm/Baeza-Yates18}.
Thus, unless otherwise specified, we use $k=10$ as a default value. Furthermore, we let the default maximum deviation $\varepsilon$ be $0.5$, aiming to strike a balance between being sufficiently close to the constraints but realistically possible in the datasets. In practice, this parameter may be chosen by specifying a worst-case scenario that is still acceptable, and then use the deviation of this scenario as calculated by \Cref{def:mospe} to set $\varepsilon$. We also set the constraints set to include a single constraint (constraint (1) from Table~\ref{tab:queriesAndConstraints} for each dataset). We used the three distance measures mentioned in \Cref{sec:distance}: the queries predicates distance measure $DIS_{pred}$ (abbr. QD in the figures), the Jaccard distance over the output, $DIS_{Jaccard}$ (JAC in the figures), and Kendall's $\tau$, $DIS_{Kendall}$, for top-$k$ lists defined in~\cite{FKS03} (KEN in the figures).

\subsubsection*{Compared algorithms}
To our knowledge, our problem is novel and has no competing algorithms other than the na\"{i}ve exhaustive search. Therefore, we compare our baseline MILP-based algorithm (MILP), our optimized MILP-based algorithm (MILP+opt), which includes the optimization described in Section~\ref{sec:optimizations}, an exhaustive search over the space of refinements (Na\"{i}ve), and a version the exhaustive search that uses our provenance annotations to evaluate the refinements (Na\"{i}ve+prov). 
We report the total running time and show the setup time (constructing the MILP for MILP-based solutions, and generating the provenance for Na\"{i}ve+prov). The MILP solver time is the gap between total and setup.
The reported times are an average of $5$ executions.

\subsubsection*{Platform \& implementation details}
Our experiments were performed on macOS 13.4 with an Apple M2 processor and 16 GB of memory. Our algorithm was implemented with IBM's CPLEX 22.1.1.0\footnote{\url{https://www.ibm.com/products/ilog-cplex-optimization-studio/cplex-optimizer}} to solve the mixed-integer linear program and DuckDB 0.8.0 \cite{DuckDB} for query evaluation. The algorithm to construct the problem and the na\"{i}ve method were written and evaluated with Python 3.9.6 and PuLP 2.7.0 (the library used for modeling the MILP problem). $DIS_{pred}$ is linearized by computing the Jaccard distance for categorical predicates through the Charnes-Cooper transformation~\cite{CC62}. In addition, for numerical predicates, additional variables are generated that represent the absolute difference between the refined and original constants. As it does not consider the output, we skip generating $s_t$ and $l_{t, k}$ variables for tuples that do not belong to any group $G$ in $\constraints$. $DIS_{Jaccard}$ is evaluated over the output, thus we leverage the fact that there are at least $k^*$ tuples in the output and aim at maximizing the number of original tuples output, thereby maximizing the Jaccard distance.
For $DIS_{Kendall}$, only Cases 2 (a tuple leaves the top-$k$) and 3 (a tuple enters the top-$k$) as defined in \cite{FKS03} may occur in our model. We create a variable for each case for each tuple, which is then equal to the sum of the case if the tuple is selected and zero otherwise. Specifically, we create variables for tuples that were present in the original top-$k^*$, as these cases pertain to pairs where at least one is present in the original output. If we have such a tuple $t$, then since we know that the order of the output is determined ahead of time by the ranking in $\widetilde{Q}$, we can sum the $l_{t', k^*}$ values for all $t'$ that match the condition of the cases. We include an upper-bound expression limiting the value to at most $0$ if the case is irrelevant (determined by $l_{t, k^*}$), and the maximum possible value otherwise. With this in mind, we model Case 2 with the expressions for each $t \in Q(D)_{k^*}$
{\allowdisplaybreaks
\begin{align*}
    CaseII_t &\leq (|\widetilde{Q}(D)| + 1) \cdot (1 - l_{t, k^*}) \\
    CaseII_t &\leq (|\widetilde{Q}(D)| + 1) \cdot l_{t, k^*} + \sum_{\mathclap{\substack{t' \in Q(D)_{k^*}\\ Q(D)(t') > Q(D)(t)}}} l_{t', k^*} \\
    CaseII_t &\geq \sum_{\mathclap{\substack{t' \in Q(D)_{k^*}\\ Q(D)(t') > Q(D)(t)}}} l_{t', k^*} - (|\widetilde{Q}(D)| + 1) \cdot l_{t, k^*}
\end{align*}
}
Similarly, we model Case 3 for each $t \ in Q(D)_{k^*}$ by the expressions
\begin{align*}
    CaseIII_t &\leq (|\widetilde{Q}(D)| + 1) \cdot (1 - l_{t, k^*}) \\
    CaseIII_t &\leq (|\widetilde{Q}(D)| + 1) \cdot l_{t, k^*} + \sum_{\mathclap{\substack{t' \notin Q(D)_{k^*}}}} l_{t', k^*} \\
    CaseIII_t &\geq \sum_{\mathclap{\substack{t' \notin Q(D)_{k^*}}}} l_{t', k^*} - (|\widetilde{Q}(D)| + 1) \cdot l_{t, k^*}
\end{align*}
Note that both $CaseII_t$ and $CaseIII_t$ have lower bounds of $0$. Finally, by minimizing the sum of these variables for each $t$ in the original output, we minimize the Kendall's $\tau$ distance for top-$k$ (over $k^*$).

\begin{figure}[t!]
    \begin{subfigure}{.23\textwidth}
      \centering
      \includegraphics[width=3.7cm]{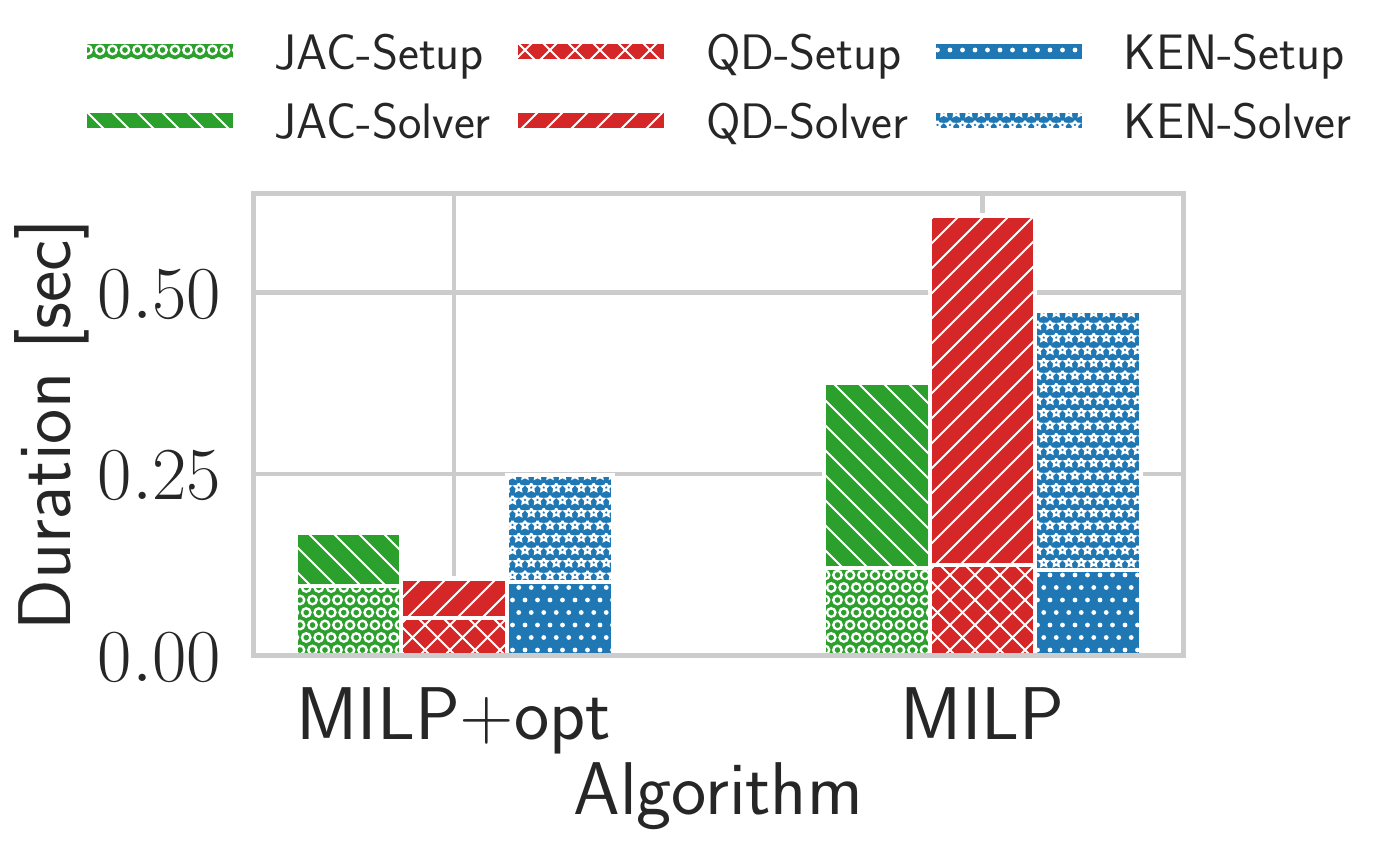}
      \caption{Astronauts}
      \label{fig:r1}
    \end{subfigure}%
    \begin{subfigure}{.23\textwidth}
      \centering
      \includegraphics[width=3.3cm]{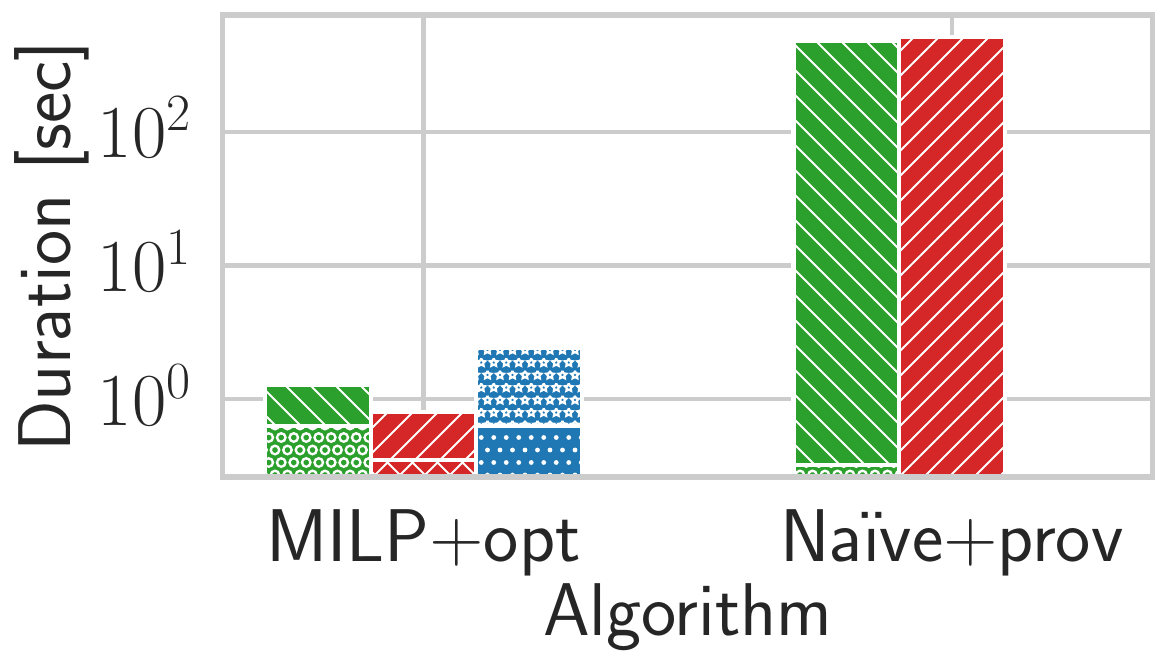}
      \caption{Law Students}
      \label{fig:law_method}
    \end{subfigure}
    \begin{subfigure}{.23\textwidth}
      \centering
      \includegraphics[width=3.4cm]{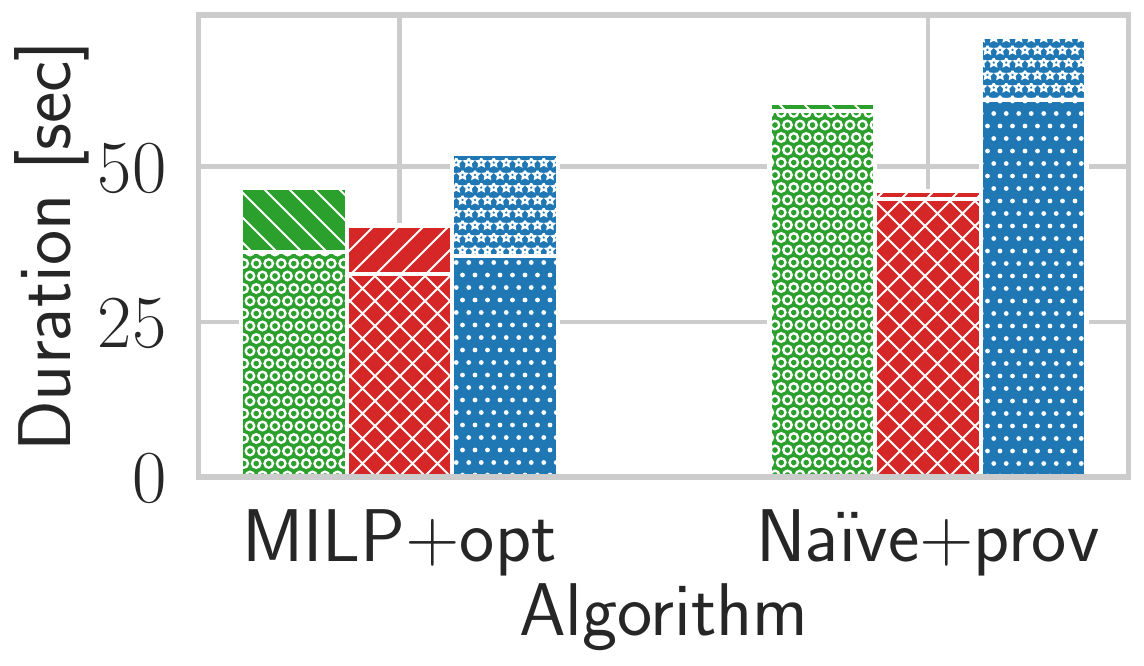}
      \caption{MEPS}
      \label{fig:meps_method}
    \end{subfigure}%
    \begin{subfigure}{.23\textwidth}
      \centering
      \includegraphics[width=3.4cm]{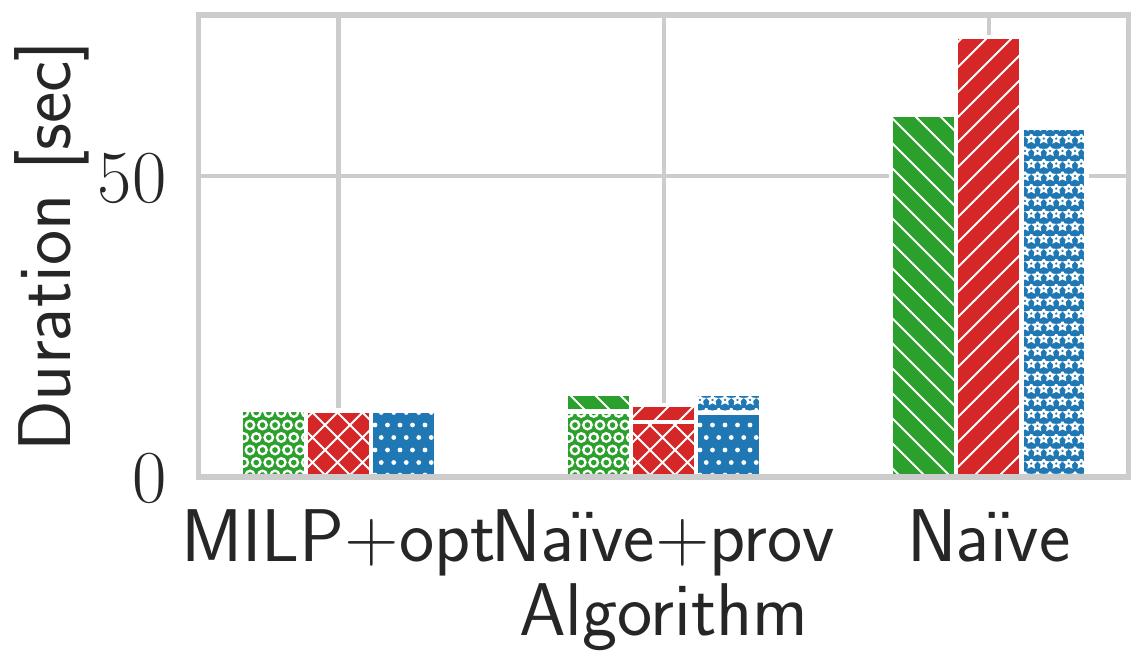}
      \caption{TPC-H}
      \label{fig:r4}
    \end{subfigure}
    \vspace{-0.3cm}
    \caption{Running time of compared algorithms, for cases where computation completed within a 1-hour timeout (method or distance omitted when timed out).  MILP+opt consistently outperforms other methods.}
    \vspace{-0.3cm}
    \label{fig:time_vs_method}
\end{figure}
\begin{figure*}[ht]
    \begin{subfigure}{.23\textwidth}
      \centering
      \includegraphics[width=4cm]{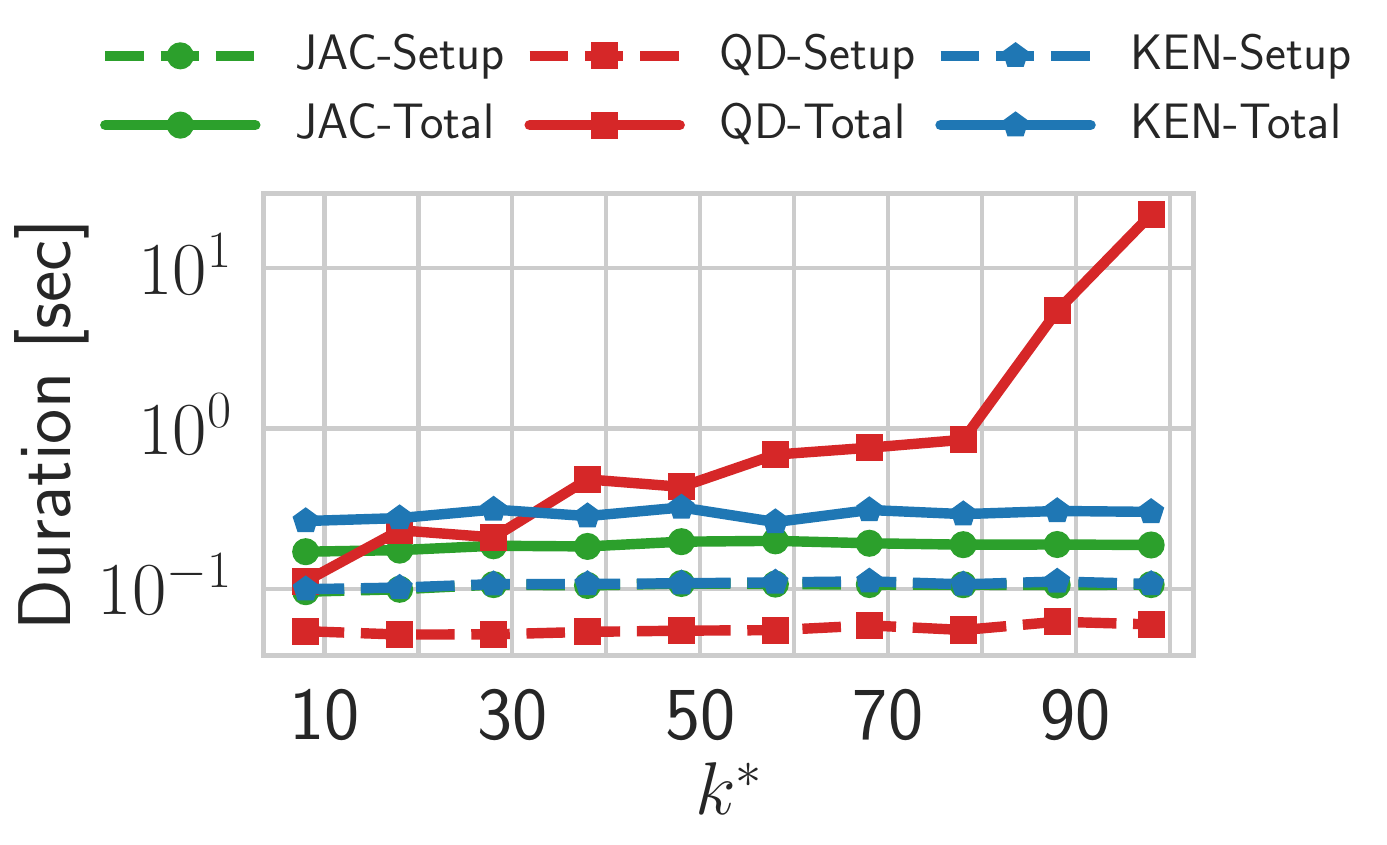}
      \caption{Astronauts (log scale)}
      \label{fig:r5}
    \end{subfigure}
    \begin{subfigure}{.23\textwidth}
      \centering
      \includegraphics[width=3.8cm]{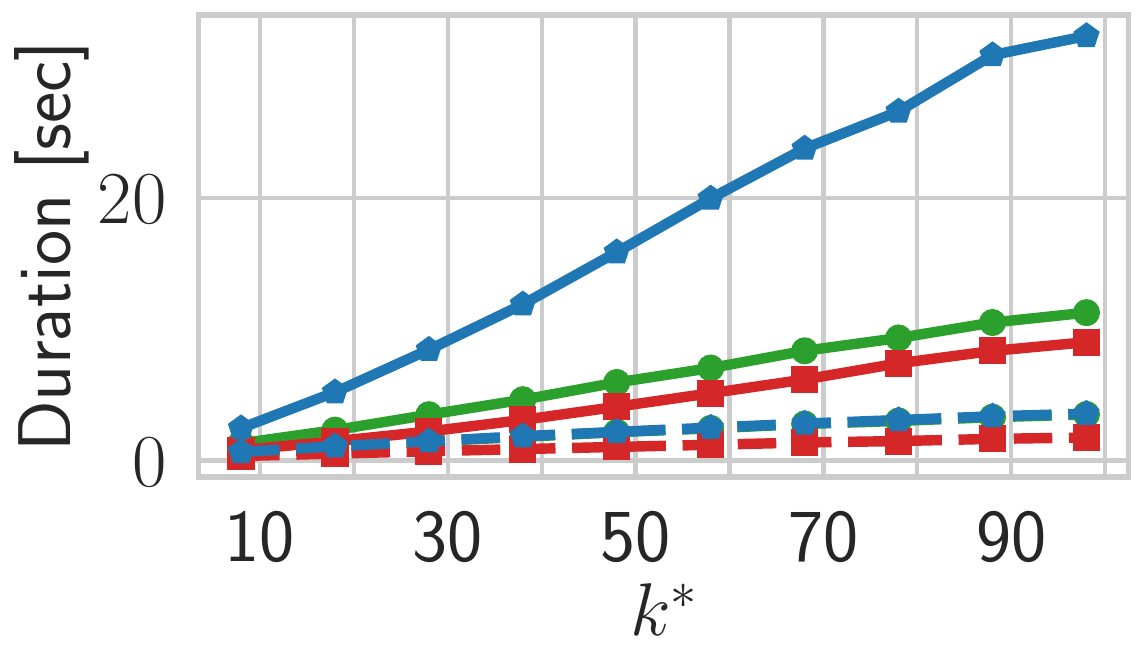}
      \caption{Law Students}
      \label{fig:r6}
    \end{subfigure}
    \begin{subfigure}{.23\textwidth}
      \centering
      \includegraphics[width=3.9cm]{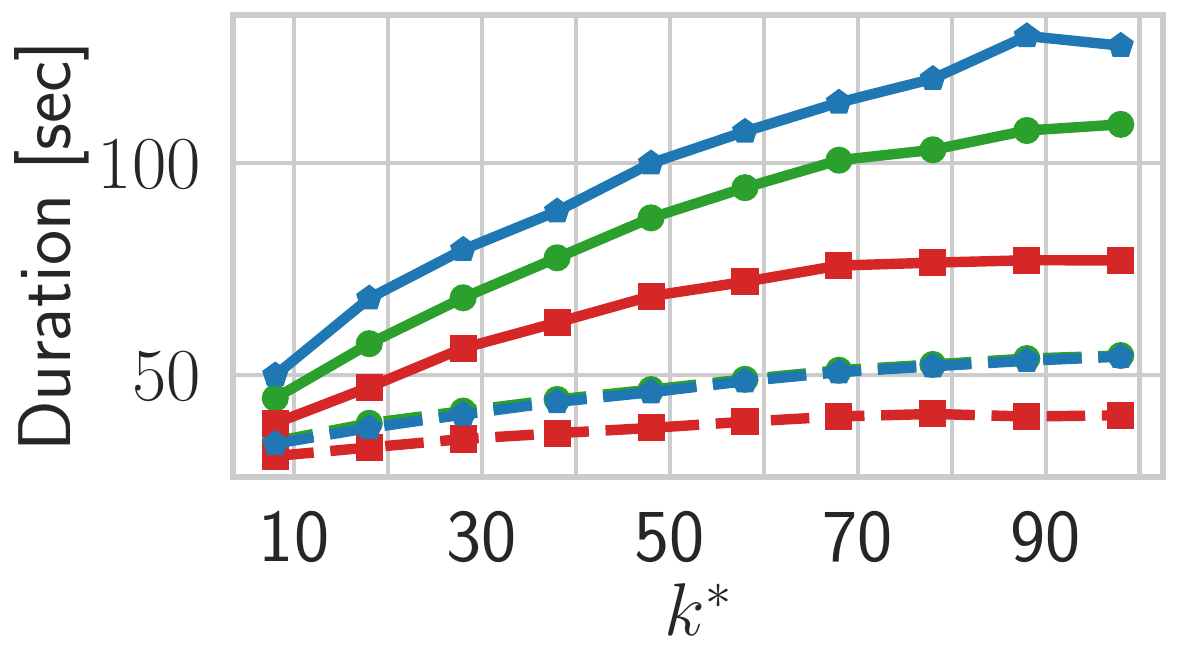}
      \caption{MEPS}
      \label{fig:r7}
    \end{subfigure}
    \begin{subfigure}{.23\textwidth}
      \centering
      \includegraphics[width=4cm]{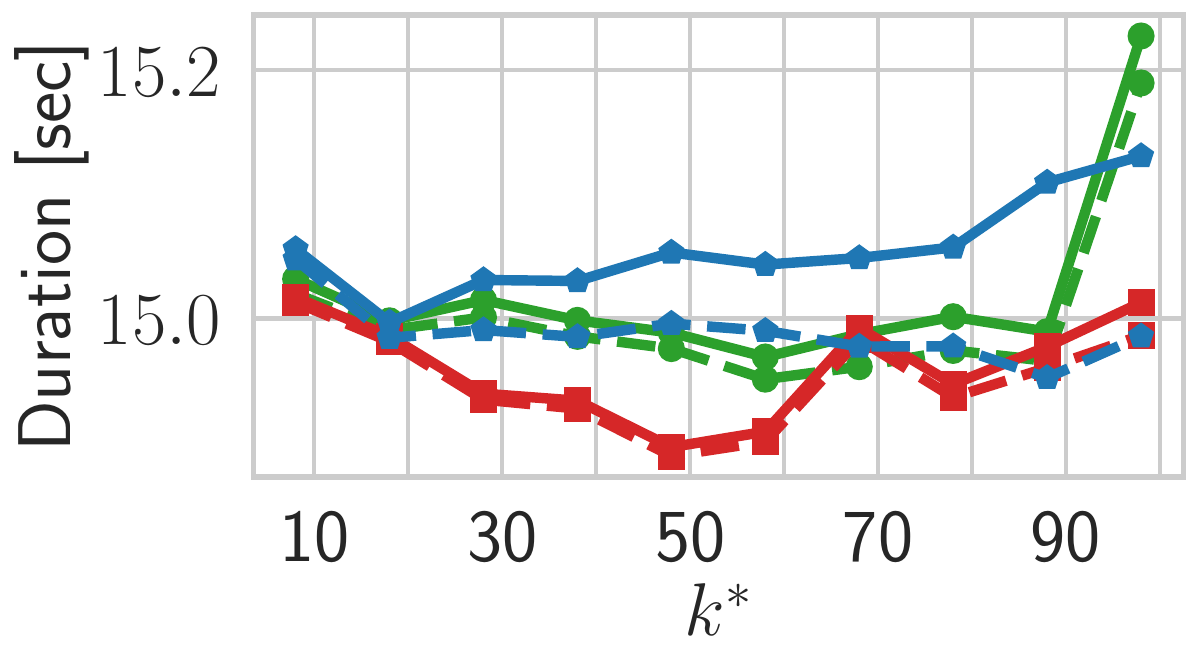}
      \caption{TPC-H}
      \label{fig:r8}
    \end{subfigure}

    \caption{Running time vs. $k^*$, showing $DIS_{pred}$ is often the fastest to compute, while $DIS_{Kendall}$ can be sensitive to increasing $k^*$.}
    \label{fig:time_vs_k}
\end{figure*}
\begin{figure*}[ht]
    \begin{subfigure}{.23\textwidth}
      \centering
      \includegraphics[width=4cm]{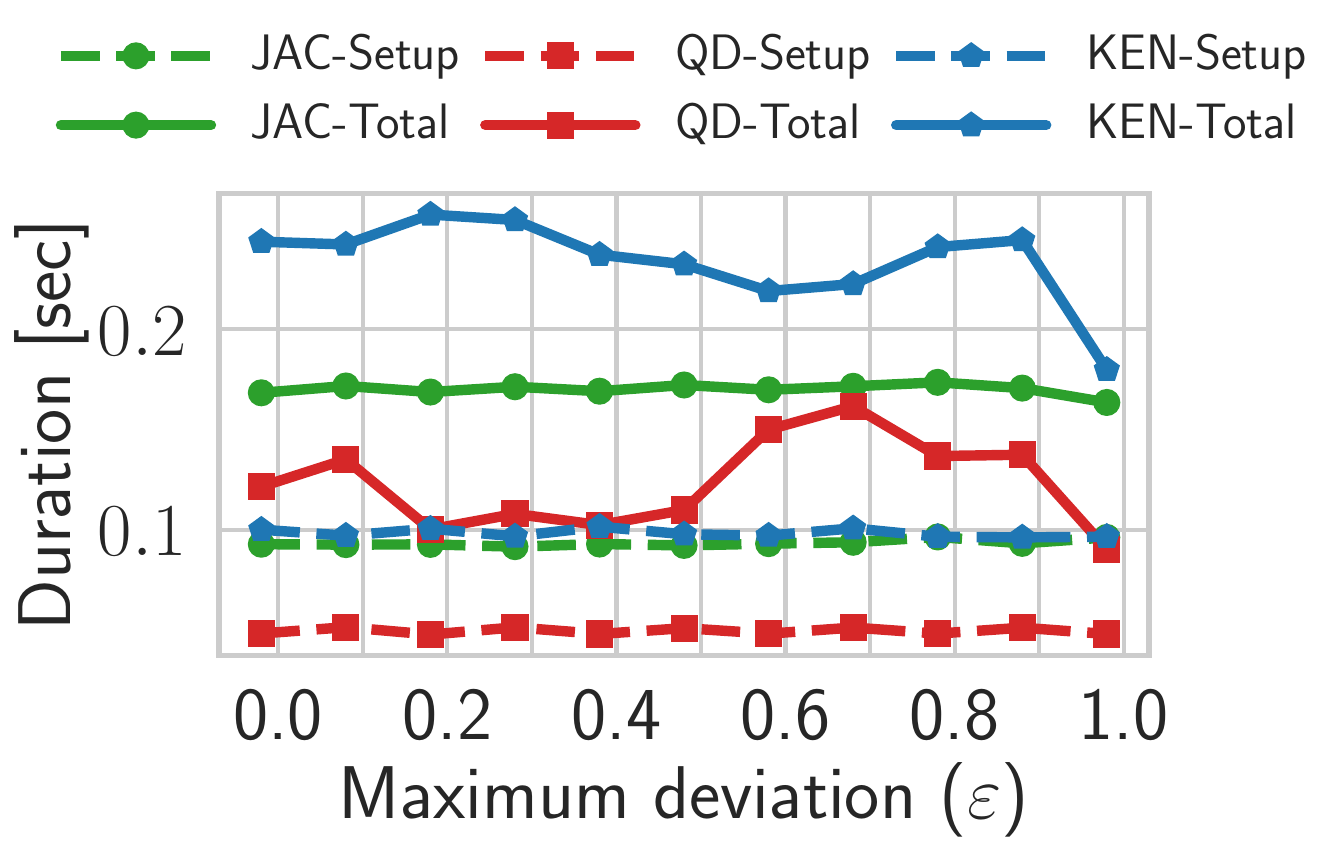}
      \caption{Astronauts}
      \label{fig:r9}
    \end{subfigure}
    \begin{subfigure}{.23\textwidth}
      \centering
      \includegraphics[width=3.7cm]{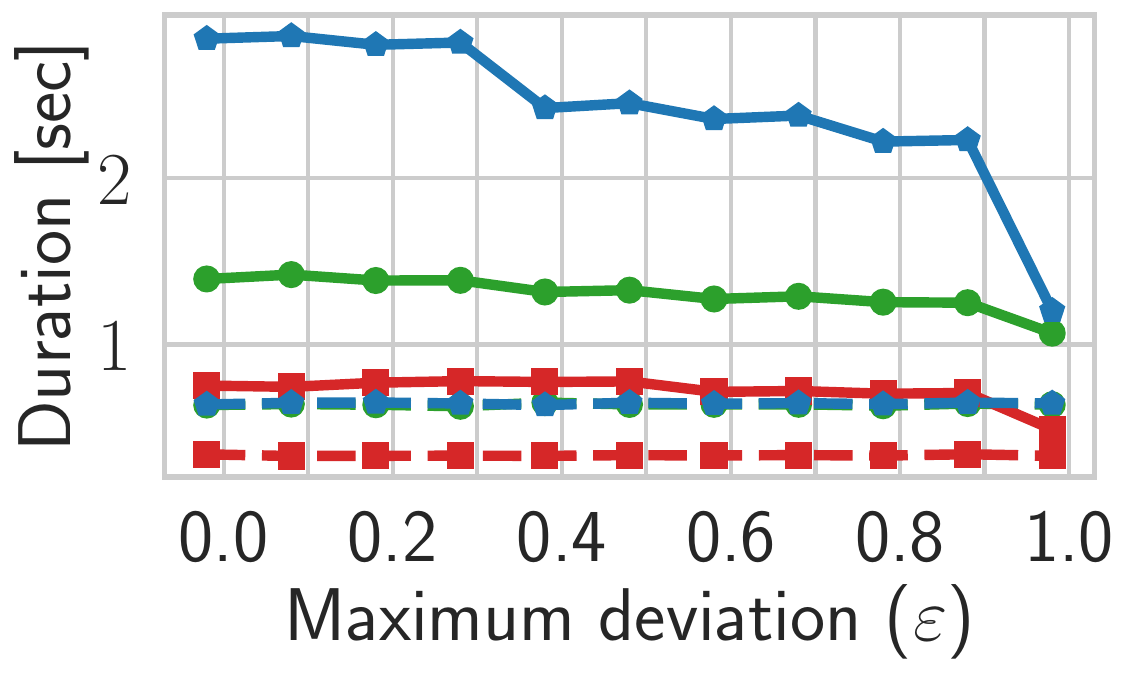}
      \caption{Law Students}
      \label{fig:r10}
    \end{subfigure}
    \begin{subfigure}{.23\textwidth}
      \centering
      \includegraphics[width=3.7cm]{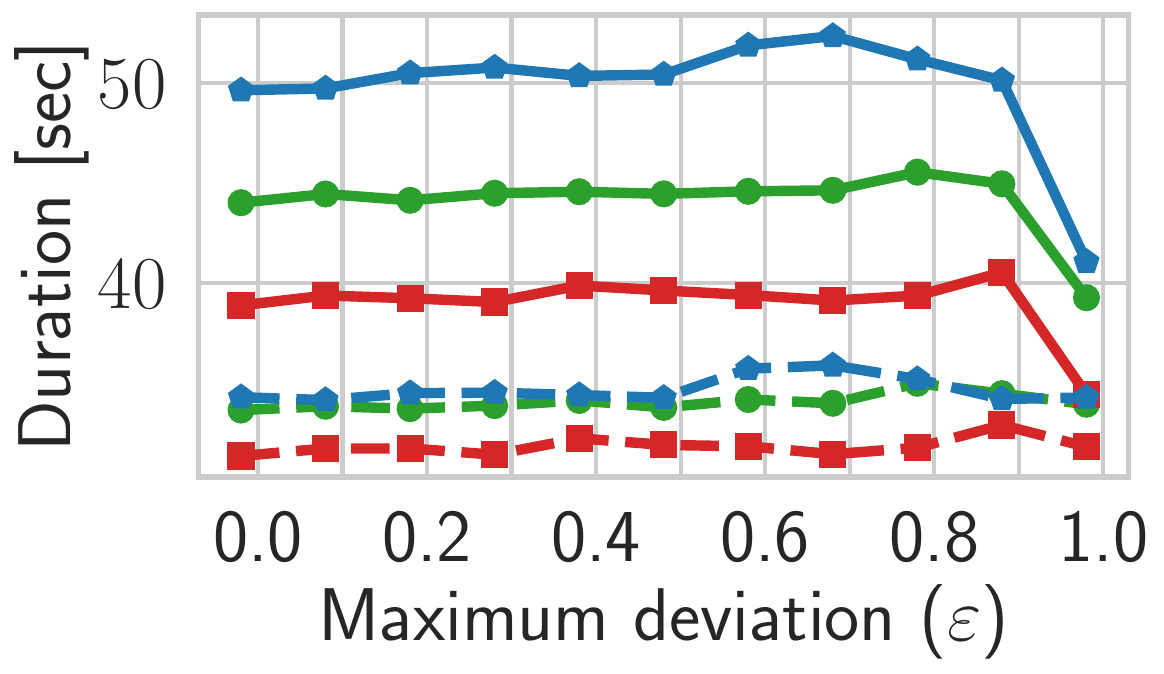}
      \caption{MEPS}
      \label{fig:r11}
    \end{subfigure}
    \begin{subfigure}{.23\textwidth}
      \centering
      \includegraphics[width=4cm]{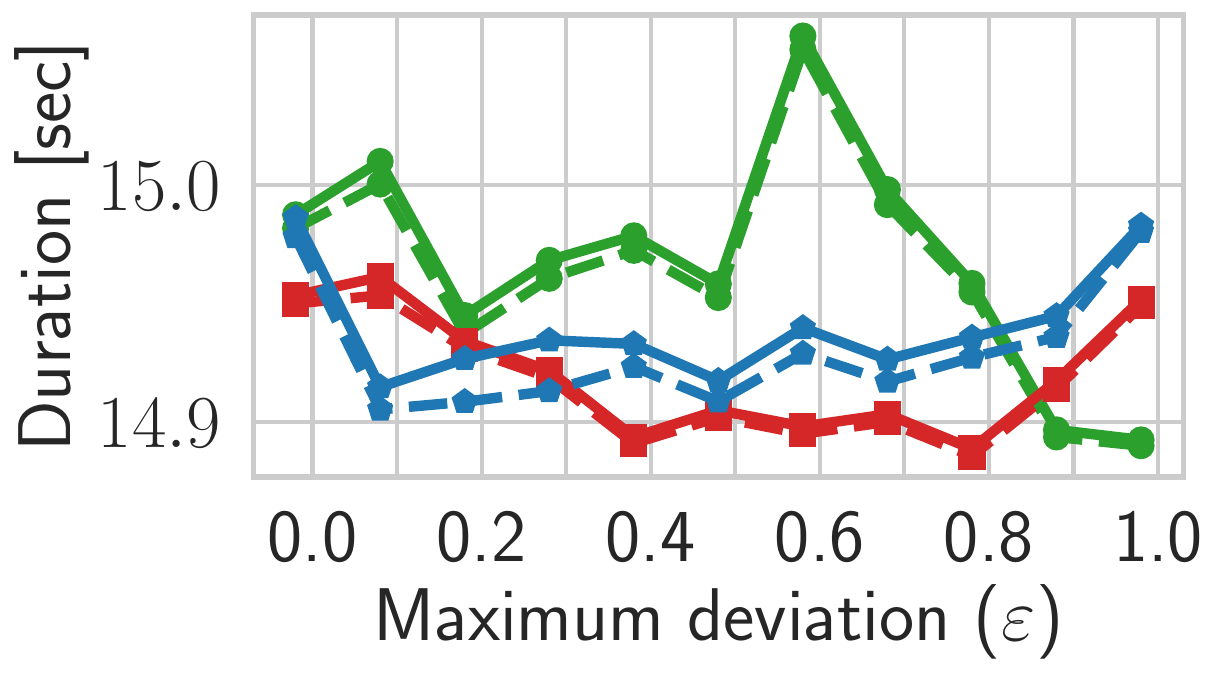}
      \caption{TPC-H}
      \label{fig:r12}
    \end{subfigure}
    
    \caption{Running time vs. maximum deviation ($\varepsilon$), showing that the effect of $\varepsilon$ is limited.}
    \label{fig:time_vs_eps}
\end{figure*}

\begin{figure*}[ht]
    \begin{subfigure}{.23\textwidth}
      \centering
      \includegraphics[width=4cm]{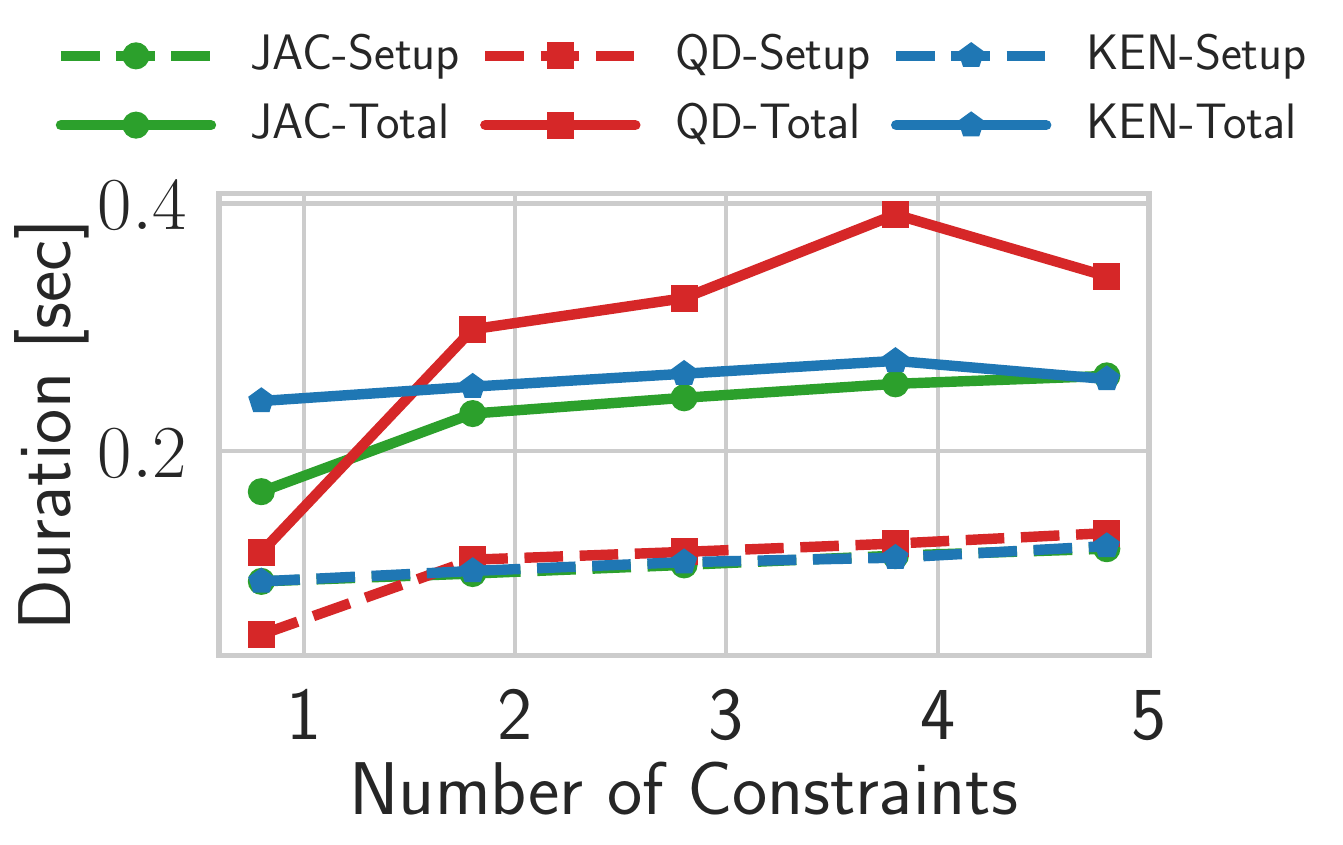}
      \caption{Astronauts}
      \label{fig:r13}
    \end{subfigure}
    \begin{subfigure}{.23\textwidth}
      \centering
      \includegraphics[width=3.7cm]{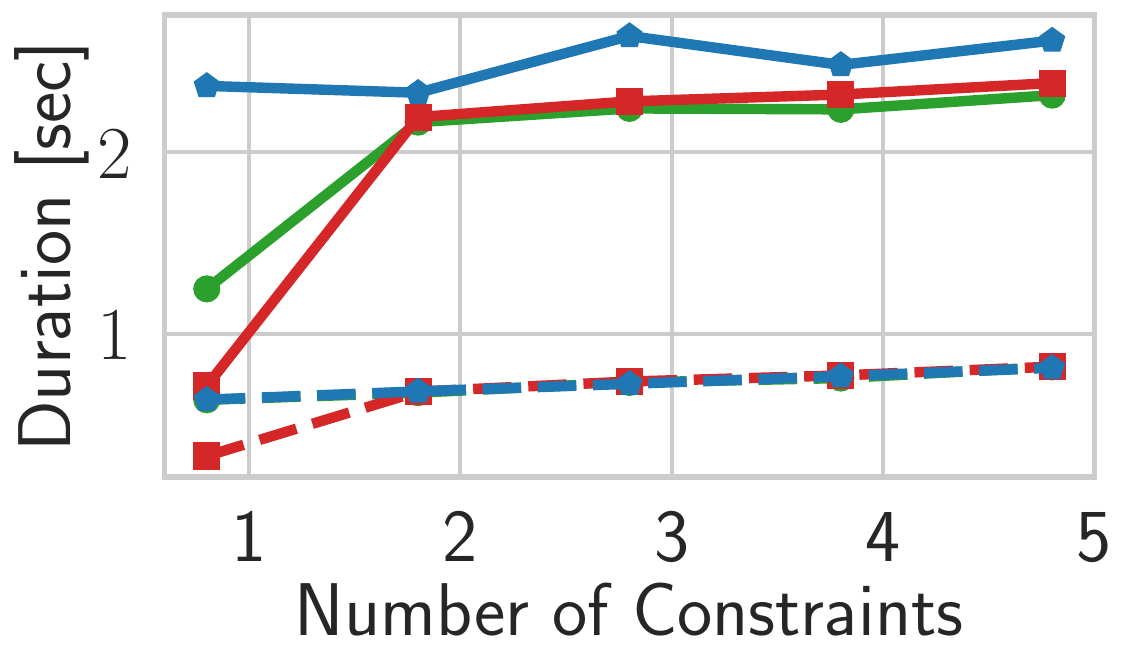}
      \caption{Law Students}
      \label{fig:r14}
    \end{subfigure}
    \begin{subfigure}{.23\textwidth}
      \centering
      \includegraphics[width=3.85cm]{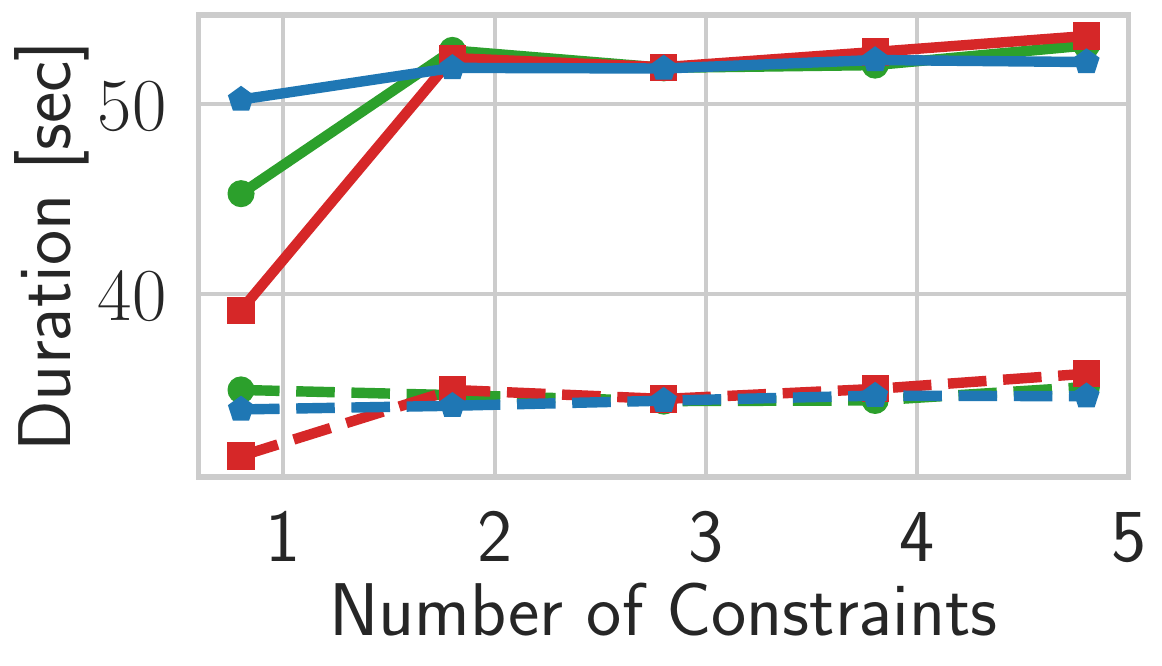}
      \caption{MEPS}
      \label{fig:r15}
    \end{subfigure}
    \begin{subfigure}{.23\textwidth}
      \centering
      \includegraphics[width=4.1cm]{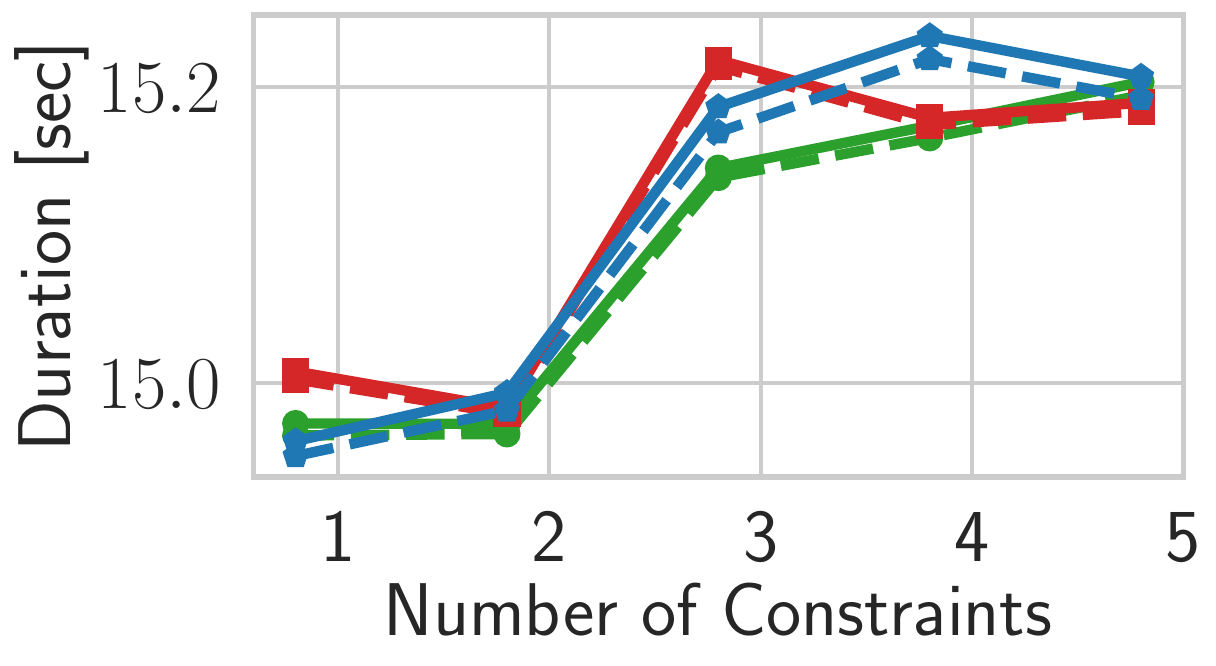}
      \caption{TPC-H}
      \label{fig:r16}
    \end{subfigure}
    
    \caption{Running time vs. the number of constraints: the impact of the number of constraints is limited.}
    \label{fig:time_vs_num_of_constraints}
\end{figure*}

\begin{figure*}[ht]
    \begin{subfigure}{.23\textwidth}
      \centering
      \includegraphics[width=4cm]{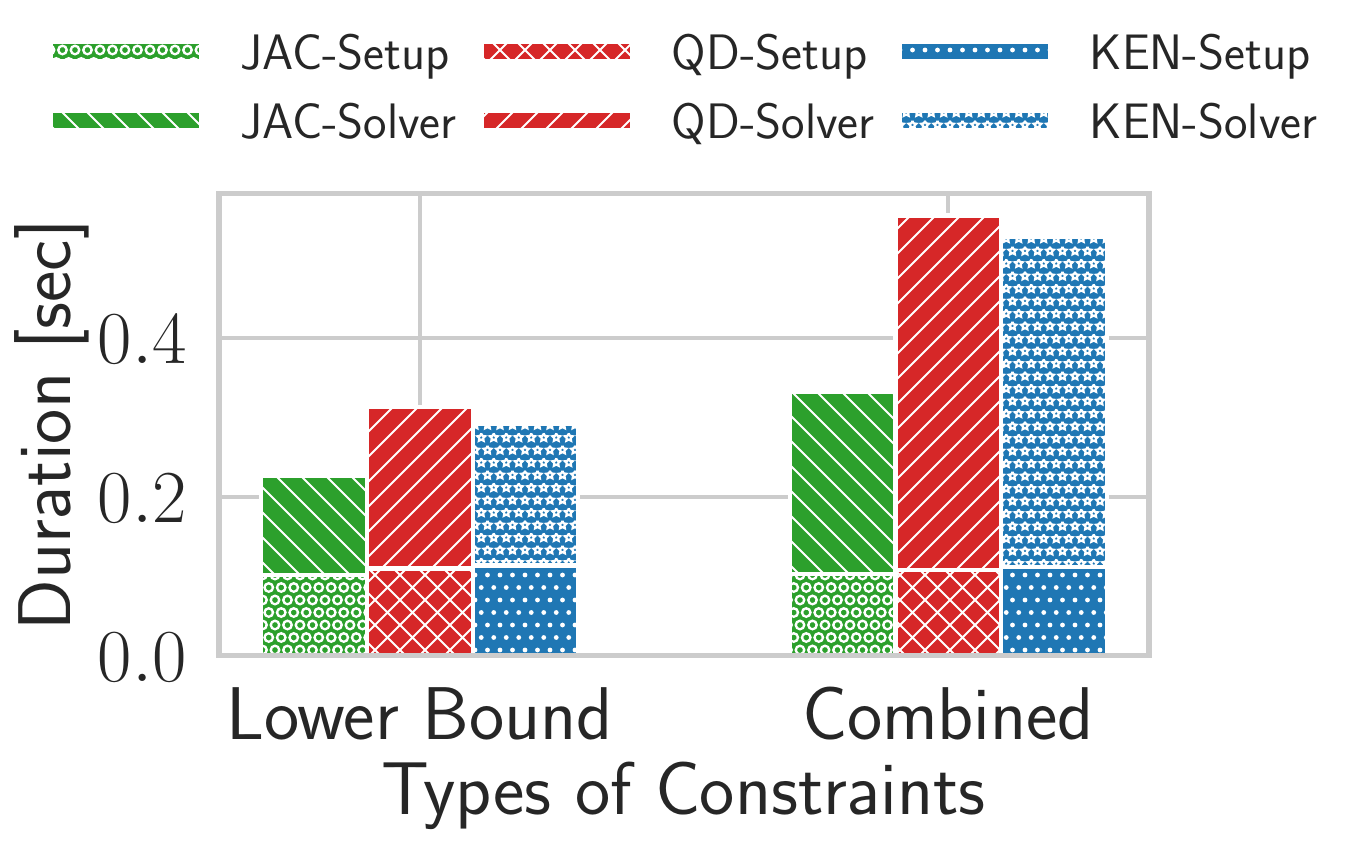}
      \caption{Astronauts}
      \label{fig:r17}
    \end{subfigure}
    \begin{subfigure}{.23\textwidth}
      \centering
      \includegraphics[width=3.4cm]{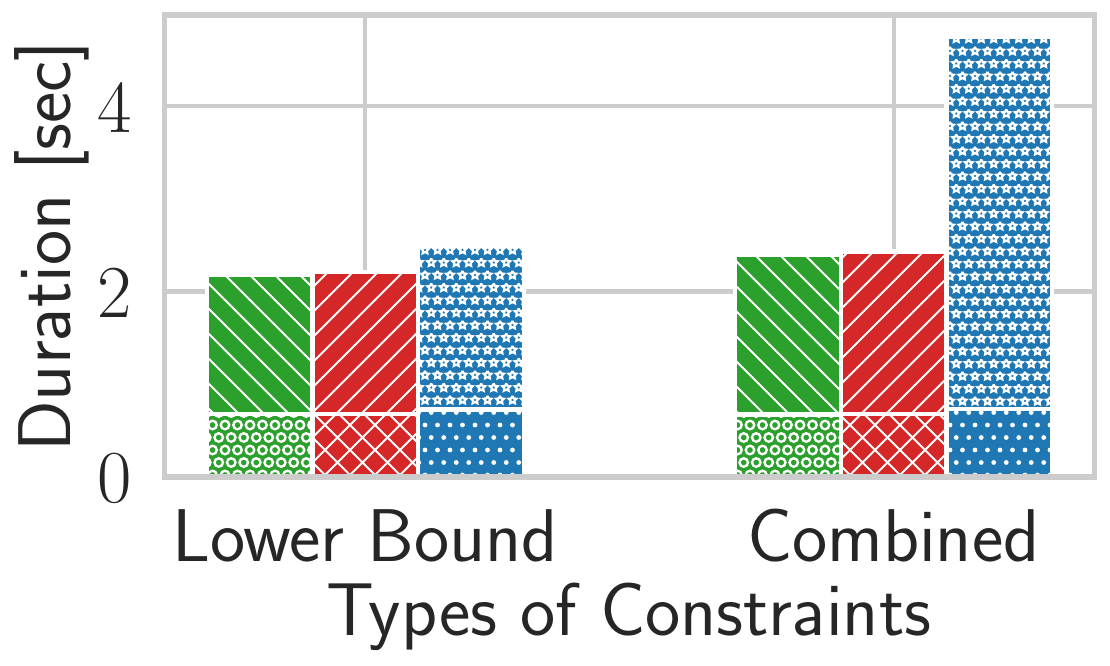}
      \caption{Law Students}
      \label{fig:r18}
    \end{subfigure}
    \begin{subfigure}{.23\textwidth}
      \centering
      \includegraphics[width=3.7cm]{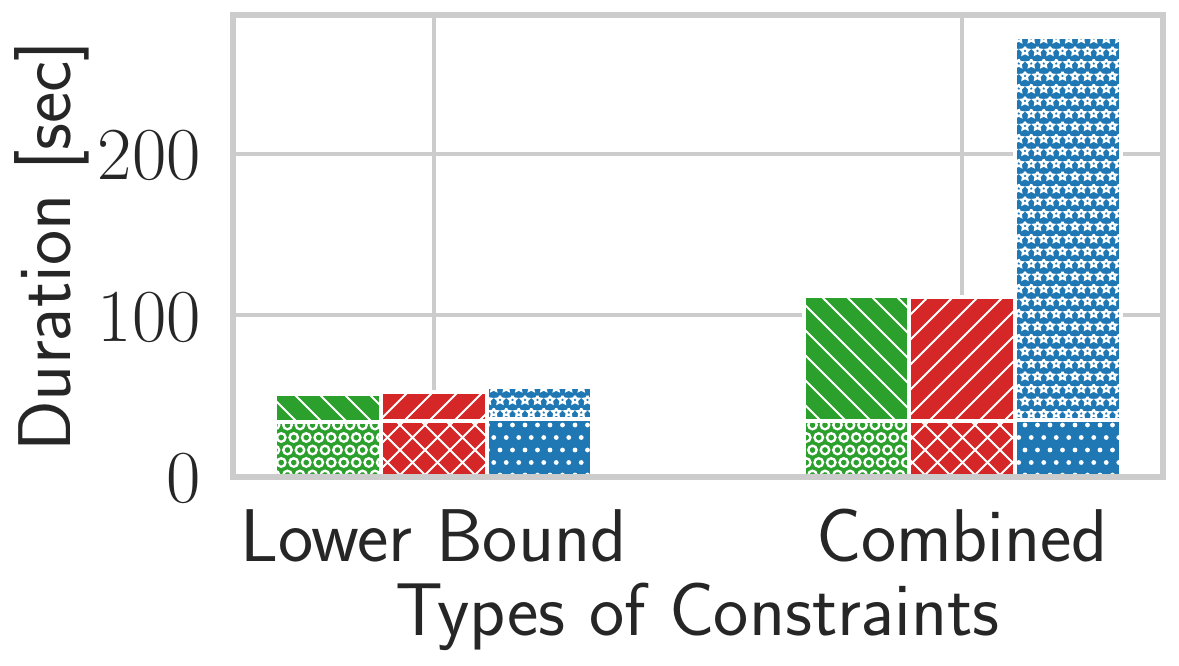}
      \caption{MEPS}
      \label{fig:r19}
    \end{subfigure}
    \begin{subfigure}{.23\textwidth}
      \centering
      \includegraphics[width=3.4cm]{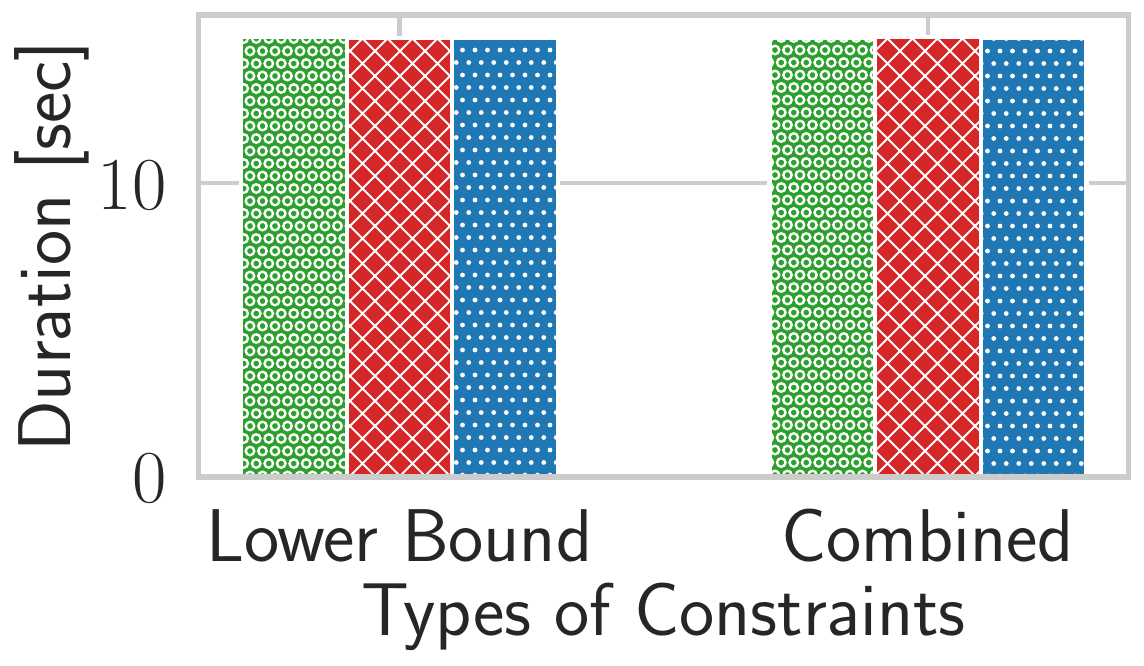}
      \caption{TPC-H}
      \label{fig:r20}
    \end{subfigure}
    
    \caption{Running time vs. constraint type, showing the efficacy of one of our optimizations.}
    \label{fig:time_vs_constraint_type}
    \vspace{-0.5cm}
\end{figure*}

\subsection{Results}\label{sec:results}

\paragraph*{Running time for compared algorithms}
We begin by comparing the running time of all algorithms using the default parameters and setting a timeout of one hour. 
Recall that the size of the generated MILP program (without optimization) is linear in the data size and that MILP solvers are typically sensitive to the program size. Thus, we expect the MILP algorithm to struggle with large-scale datasets. On the other hand, the na\"{i}ve approaches perform a brute-force search over the possible refinements, where their number is exponential in the number of predicates in the query (and their domain). Thus, datasets with high cardinality in the domain of the query predicate are likely to be challenging for the na\"{i}ve solutions.  

Figure~\ref{fig:r1} presents performance for the Astronauts dataset. The optimized MILP solution outperforms the unoptimized MILP, and we observe a speedup of up to 6 times. Given that there are $114$ different values for ``Graduate Major'', the space of refinements is extremely large and both Na\"{i}ve and Na\"{i}ve+prov timed out (and thus omitted from the graph). \Cref{fig:law_method,fig:meps_method,fig:r4} show the results for the remaining datasets. In these cases, due to the data size, the unoptimized MILP was unable to terminate before the time-out. MEPS and TPC-H have a relatively small space of refinements for the posed queries, making Na\"{i}ve+prov competitive with MILP+opt. However, Law Students has a considerably larger space of refinements (although modest compared to Astronauts), making Na\"{i}ve time out and Na\"{i}ve+prov significantly slower than MILP+opt. Essentially, MILP+opt is well-posed to deal with scaling both the data size and the space of the possible refinements.
The na\"{i}ve brute-force search methods and unoptimized MILP method fail to scale, and we exclude them from the rest of the experiments.

\paragraph*{Effect of $k^*$}
We study the effect of $k^*$, the largest $k$ with a constraint in the constraint set, on the running time of our algorithm by increasing the parameter $k$ of the constraint from $10$ to $100$ in increments of $10$. The results are presented in Figure~\ref{fig:time_vs_k}. Recall that the relevancy-based optimization from \Cref{sec:optimizations} aims at reducing the program size using $k^*$. 
We expect to see its effect degrade as $k^*$ increases, as shown in \Cref{fig:r6,fig:r7}. 
The optimization is less effective for  Astronauts (Figure~\ref{fig:r5}), as the number of different lineage equivalent classes is large, and each consists of a relatively small number of tuples (fewer than $10$). 
Therefore, the expression generated for very few tuples may be removed from the program. The optimization is particularly effective for $Q_5$ of TPC-H (\Cref{fig:r8}), as the vast majority of expressions are removed as there are only $5$ lineage-equivalent classes. Moreover, we see that most of the time is spent setting up the problem rather than solving it.

\paragraph{Maximum deviation ($\varepsilon$)}
\label{sec:time_vs_eps}
While an increase in $\varepsilon$ may make finding feasible refinements easier, the solver must still find the minimal refinement, which remains a difficult task. Therefore, the value of $\varepsilon$ should not significantly affect the running time. Figure~\ref{fig:time_vs_eps} shows that the running time is fairly stable. We observed a decrease when $\varepsilon$ reaches $1.0$. This is because we use only lower-bound constraints in this experiment, where the deviation of any (refined) query is bounded by $1.0$, i.e. finding a satisfying refinement is trivial as all refinements are good enough.
In \Cref{fig:r12}, 
the solver time is negligible and depends mostly on the setup time, which is very similar across all values of $\varepsilon$ (differing by at most $1\%$).

\paragraph*{Number of constraints}
The number of generated expressions of the form (\ref{eq:in_prefix_inline}) and (\ref{eq:tuples_to_satisfy_inline}) is linear in the number of constraints. Thus, when increasing the number of constraints, the program size increases and as a result, we expect to see an increase in the running time.
We gradually added constraints to the constraint set in the order they listed in Table~\ref{tab:queriesAndConstraints}. To ensure that the set of constraints can be satisfied along with the default value of $\varepsilon$, we slightly adjust the value of the first two constraints for Astronauts, Law Students, and MEPS to have a lower-bound of $\frac{k}{3}$. As shown in Figure~\ref{fig:time_vs_num_of_constraints}, we observed a slight increase in the running time as the number of constraints increased in contrast to increasing the value of $k$. The number of expressions is linear in the number of constraints and tuples, however there are significantly fewer constraints than tuples, which is why the number of constraints does not have a pronounced effect on the runtime.
TPC-H \Cref{fig:r16} shows a negligible difference as the vast majority of the time is set up the MILP problem, as the solver has only $5$ lineage equivalence classes to explore.

\paragraph*{Effect of constraints type}
In \Cref{sec:optimizations}, we presented an optimization that is effective when tuples belong to groups with only either lower-bound or upper-bound constraints made on them. To demonstrate the effect of this optimization, we generate two sets of constraints for each dataset: $\constraints_{L}$ with lower bound constraints only, and $\constraints_{M}$ with a mixed set of upper bound and lower bound constraints. In particular, each dataset, $\constraints_{L}$ includes constraints (1) and (2) from Table~\ref{tab:queriesAndConstraints}, and $\constraints_{M}$ includes constraints (1) and (2), where constraint (2) is turned into an upper-bound constraint. 
Notice that these particular attributes are binary, and, as we assume there are at least $k^*$ tuples in the output, they are equivalent (except for TPC-H, which lacks any binary attributes). We then compared the running time when using $\constraints_{L}$ and $\constraints_{M}$. The results are presented in \Cref{fig:time_vs_constraint_type}. As expected, the running times for the case of $\constraints_{L}$ are typically better, as shown in \Cref{fig:r17,fig:r18,fig:r19}, indicating the usefulness of the optimization. We note that the experiment in \Cref{fig:r20} depicting the experiment for TPC-H shares the same performance characteristics as the previous experiments.

\begin{figure*}[ht]
    \begin{subfigure}{.23\textwidth}
      \centering
      \includegraphics[width=4cm]{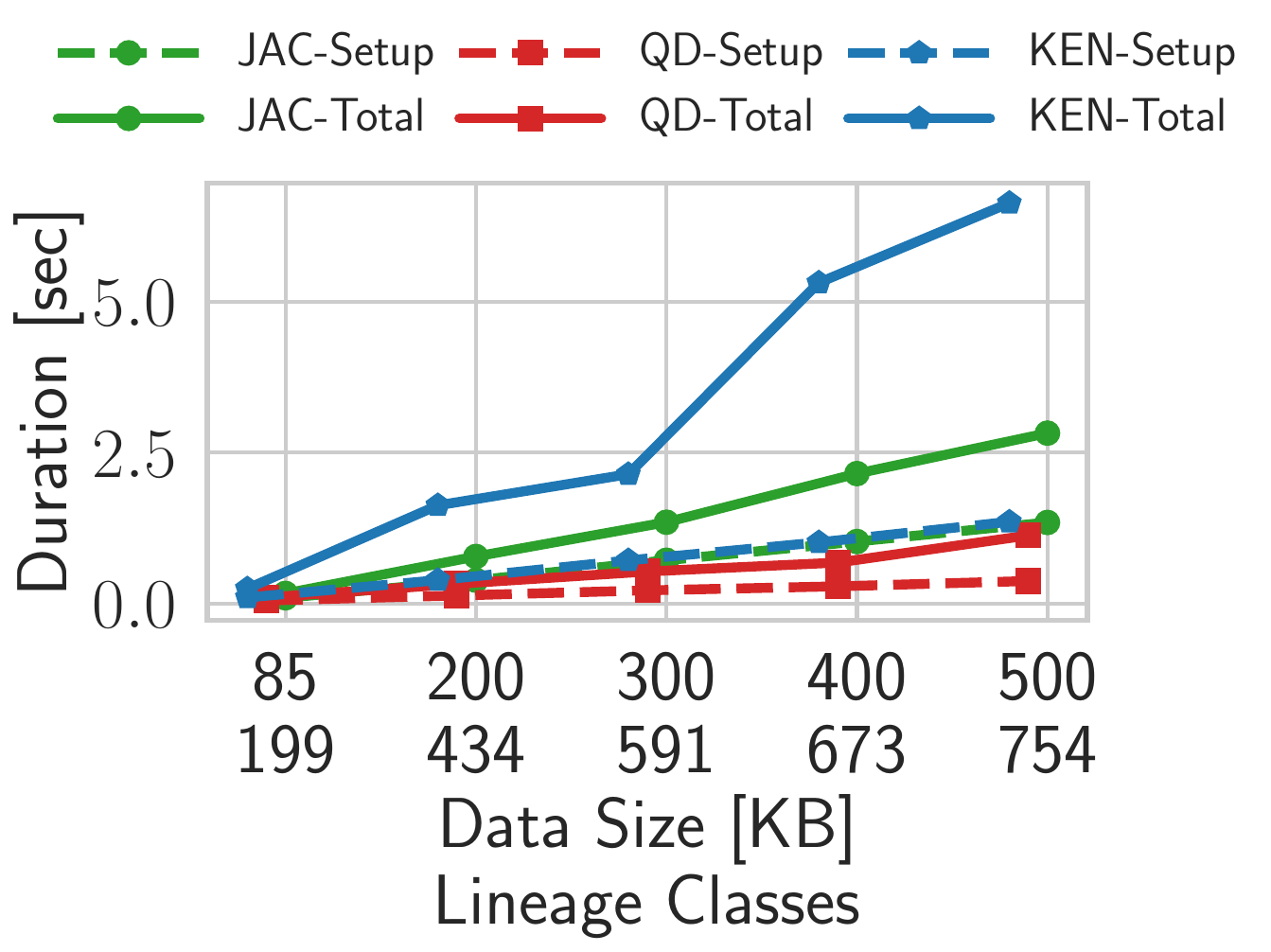}
      \caption{Astronauts}
      \label{fig:r21}
    \end{subfigure}
    \begin{subfigure}{.23\textwidth}
      \centering
      \includegraphics[width=3.6cm]{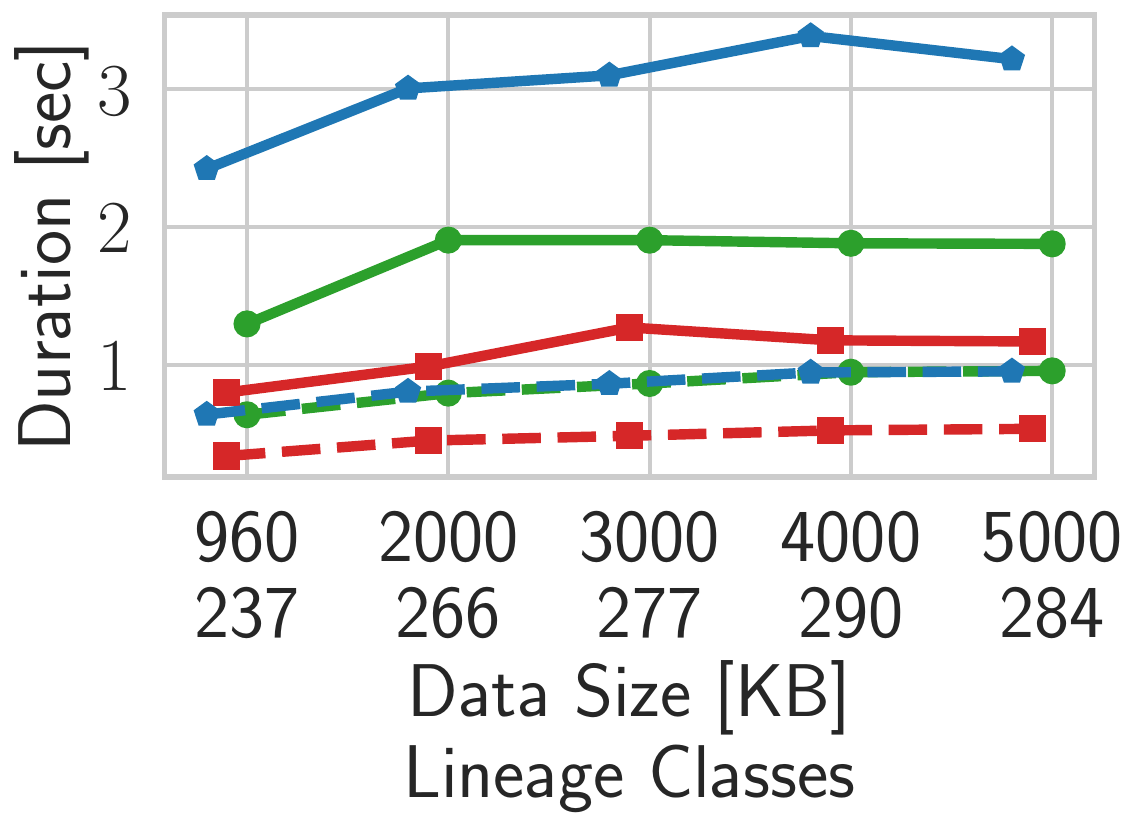}
      \caption{Law Students}
      \label{fig:r22}
    \end{subfigure}
    \begin{subfigure}{.23\textwidth}
      \centering
      \includegraphics[width=3.7cm]{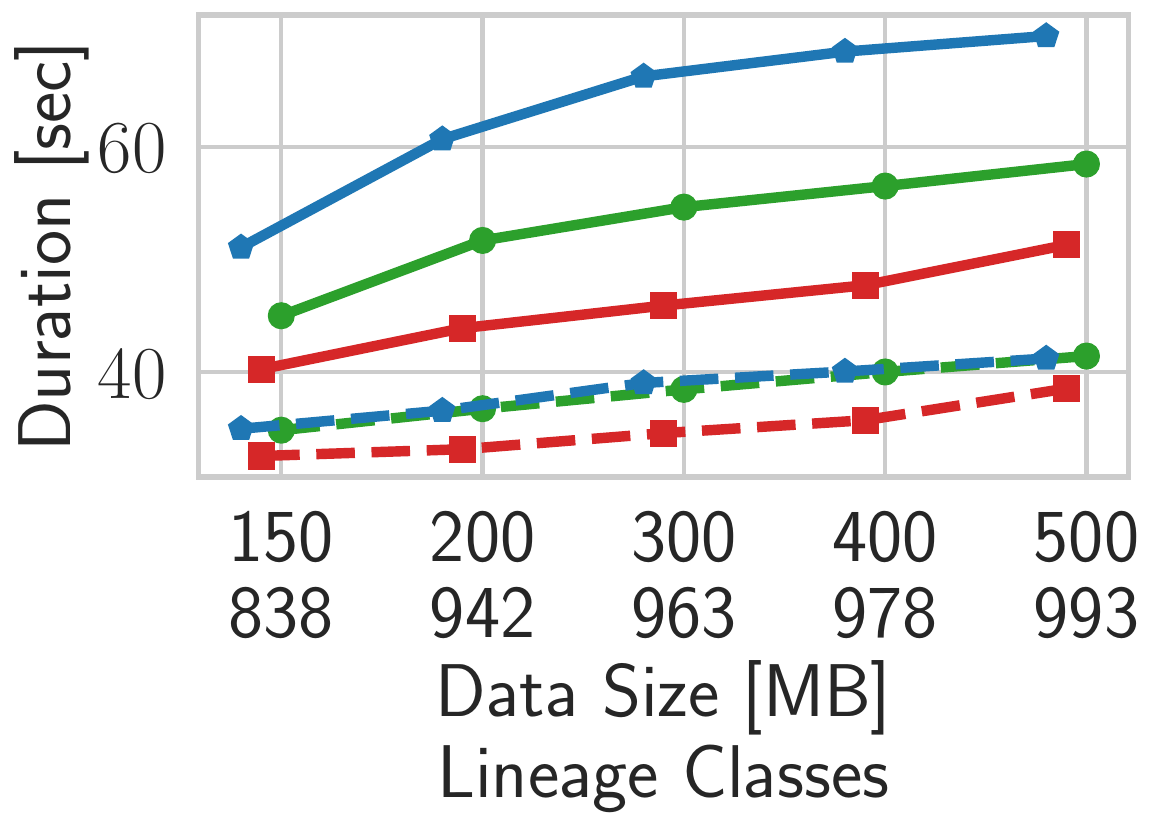}
      \caption{MEPS}
      \label{fig:r23}
    \end{subfigure}
    \begin{subfigure}{.23\textwidth}
      \centering
      \includegraphics[width=3.7cm]{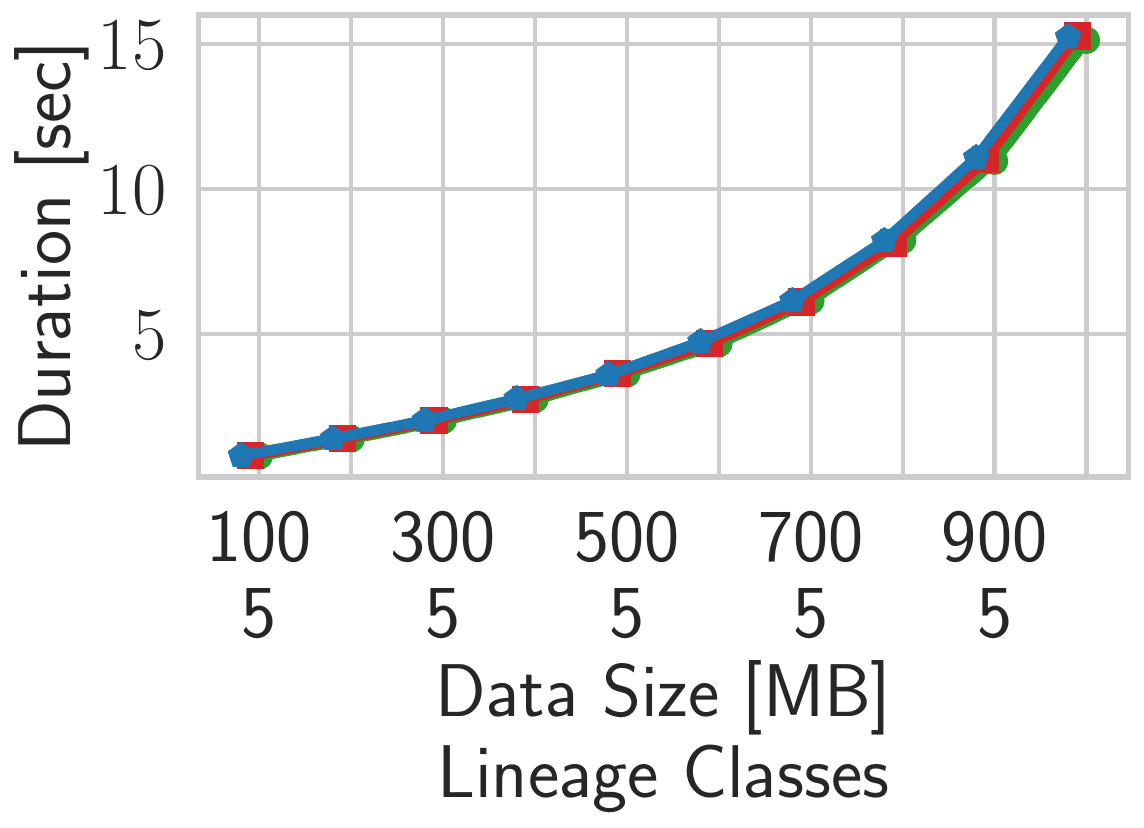}
      \caption{TPC-H}
      \label{fig:r24}
    \end{subfigure}
    
    \caption{Running time vs. data size. The setup time is mostly impacted by the cost to capture lineage from the input query, while the solving time is mostly impacted by the number of lineage classes.}
    \label{fig:time_vs_size}
    \vspace{-0.5cm}
\end{figure*}

\paragraph*{Effect of dataset size}
We use SDV~\cite{SDV} to synthesize scaled-up versions of the real datasets. Not only does this increase the data size, but new lineage classes are created according to the distribution of the dataset as well. For TPC-H, we generate different scales of the dataset according to its standard, but no new lineage classes are created. The number of variables and expressions of the generated MILP is linear in the number of tuples in the dataset. However, given that solving MILP is in ${\sf NP}$, we expect a non-polynomial increase in running time with an increase in the data size. 
The results are plotted in \Cref{fig:time_vs_size}, each plot starting from the original size of the dataset. For the Astronauts, Law students, and MEPS datasets, we observed a modest increase in the runtime as the data size grows. This could be explained by the low increase in the number of lineage classes, which impacts the efficiency of our optimization and has a greater effect on the running time.
In TPC-H (Figure~\ref{fig:r24}), the vast majority of the running time is spent building the MILP problem, and in this case, constructing the set of lineages for Q5 involves a non-trivial amount of join processing.

\paragraph*{Distance measure}
We observed that in most cases, $DIS_{Kendall}$ is the hardest to compute as it involves extra variables in order to linearize the measure. When the refinement space is extremely large, such as for Astronauts,  $DIS_{pred}$ takes longer to prove optimality (as seen in \Cref{fig:r5,fig:r13}).

\begin{figure}[t!]
    \begin{subfigure}{.23\textwidth}
      \centering
      \includegraphics[width=4.3cm]{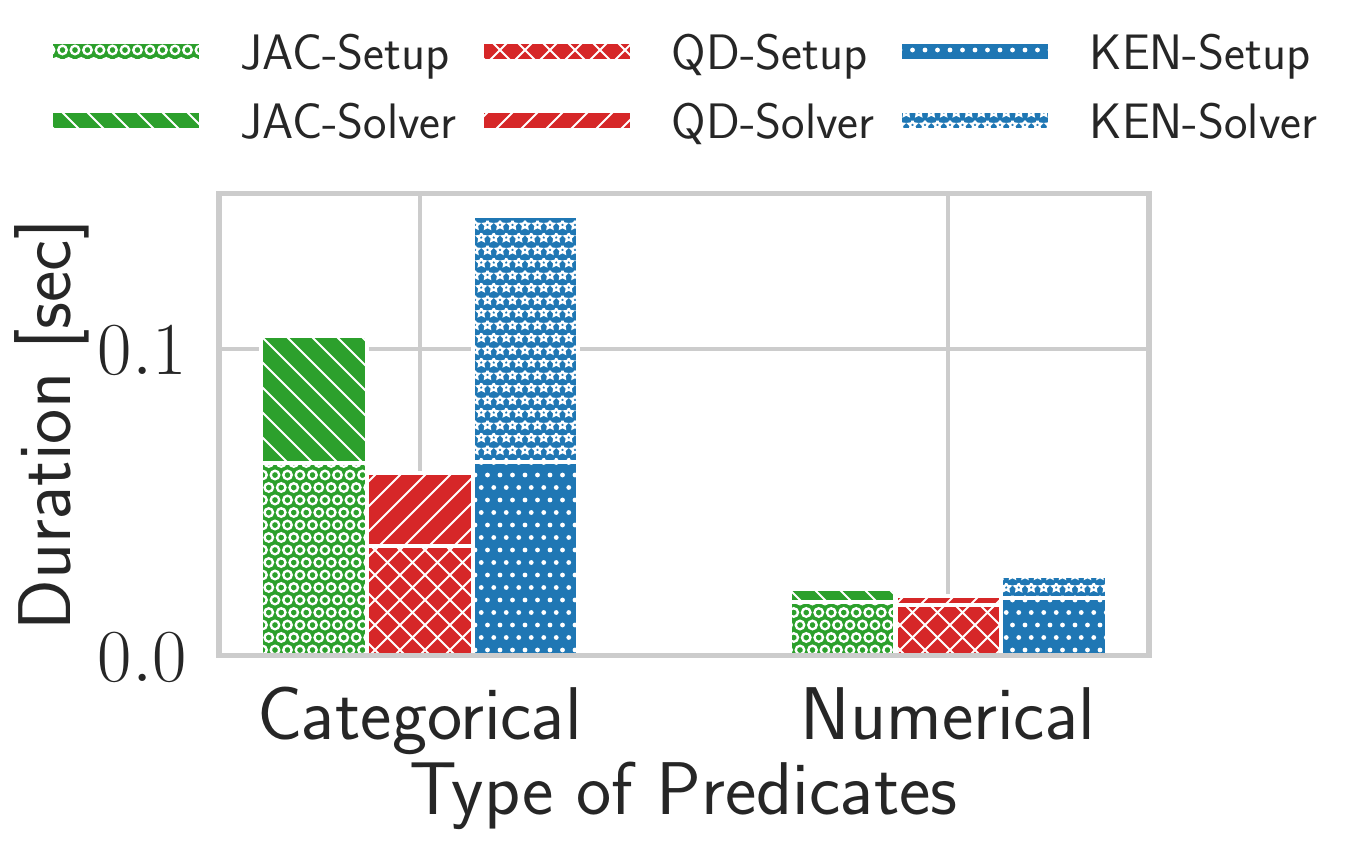}
      \caption{Astronauts}
      \label{fig:p1}
    \end{subfigure}%
    \begin{subfigure}{.23\textwidth}
      \centering
      \includegraphics[width=4cm]{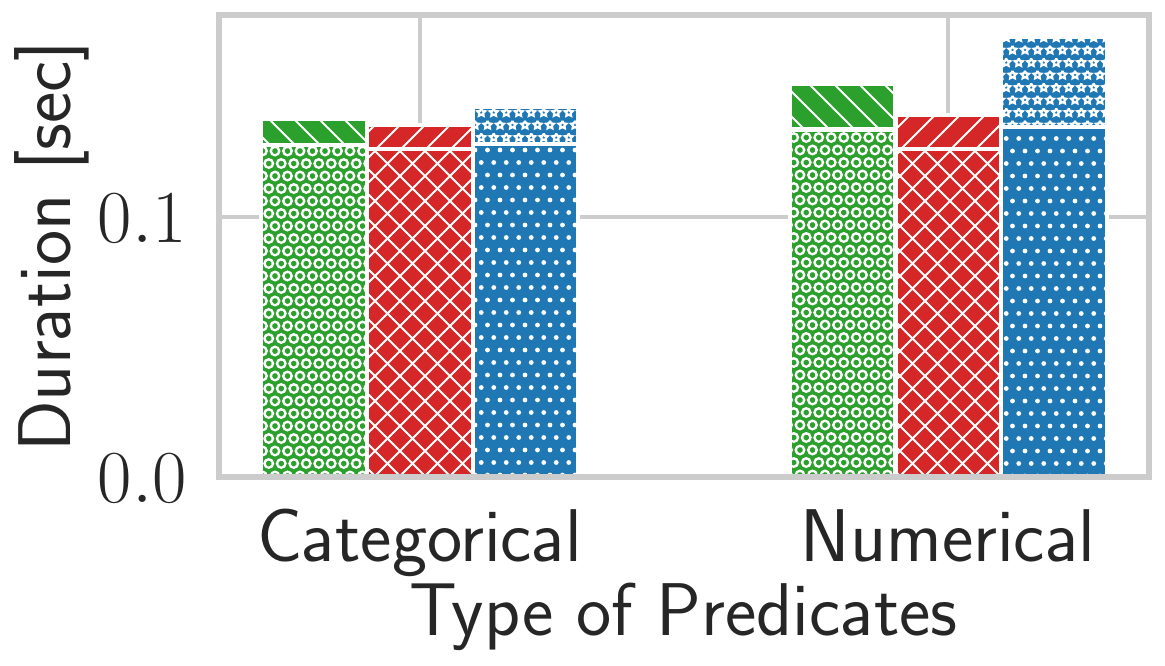}
      \caption{Law Students}
      \label{fig:p2}
    \end{subfigure}
    \vspace{-0.3cm}
    \caption{Running time as a function of the type of predicates used in the query, showing that categorical predicates may take more time to refine if not negligible.}
    \vspace{-0.3cm}
    \label{fig:time_vs_pred_type}
\end{figure}
\paragraph*{Effect of predicates type}
Recall from \Cref{sec:search} that categorical and numerical predicates are processed differently by our MILP model. While categorical values map directly to one MILP variable, a few auxiliary variables and expressions must be created in order to model the semantics of numerical predicates. Therefore, we are interested in whether or not this handling makes numerical predicates less efficient to refine than categorical. The queries defined for MEPS and TPC-H lack numerical predicates, so we use Astronauts and Law Students for this experiment. We then compare the runtime of refining $Q_A$ and $Q_L$ with either only their categorical or numerical predicates. Recall that the categorical attribute of Astronauts has a very large domain, so the number of refinements of this predicate only is extremely large. Therefore, as shown in \Cref{fig:p1}, we see that refining $Q_A$ with only its categorical predicate takes longer than $Q_A$ with only its numerical predicates. \Cref{fig:p2} shows the difference is negligible -- this follows since the space of refinements is much smaller given the smaller domain for its categorical attribute. We conclude from this experiment that in some cases, categorical predicates can take more time to refine if the size of their domain is large.

\subsection{Comparison with Erica \cite{ERICA,ERICAfull}}
\label{sec:erica_comparison}
We conclude with a comparison to 
Erica~\cite{ERICA,ERICAfull}, which presents a similar framework for query refinements to satisfy cardinality constraints over groups representation in the query's output. We note Erica focuses on cardinality constraints over the {\it entire} output, without considering the order of tuples. By restricting the overall output size to $k$, Erica may be used to refine a given query to satisfy constraints over the top-$k$ tuples. However, as we next demonstrate, this additional constraint over the output size also limits the possible refinements to those that have at most $k$ tuples. Moreover, this adjustment of Erica cannot be used to constrain over different values of $k$ simultaneously (as in our running example). Additionally, since satisfying the constraints in our setting is more challenging, we focus on finding approximate solutions that are close to satisfying the constraints, while Erica only finds solutions that satisfy the constraints exactly. Finally, our framework allows the user to define different distance measures between queries whereas Erica uses a single distance measure based on the predicate distance.
We compare the systems by refining the query $Q_L$ except with the predicates {\tt Region = `GL' AND GPA >= 3.0}.
subject to the singleton constraint set $\constraints = \{ \lb{Sex='F'}{k=100} = 50 \}$. To be consistent with~\cite{ERICA,ERICAfull}, we aim at minimizing the predicate distance (using $DIS_{pred}$ as the distance measure) and allow only results that satisfy the constraints exactly, i.e., $\varepsilon = 0$. When running Erica, we added a constraint requiring that exactly $100$ results are returned to ensure the top-$100$ tuples contains at least $50$ female candidates, and that there are enough results to satisfy the assumption in our problem definition. 
Using our optimized MILP-based approach, we were able to find a minimal refinement in $\approx 11$ seconds. The refinement selects candidates from the regions {\tt `GL'} or {\tt `SC'} with a GPA of at least $4.0$. Erica found $5$ different refinements in $\approx 53$ seconds, none of them are closer to $Q$ than the refinement found by our framework. In fact, all of the refinements require a GPA of at least $4.0$ and select {\it 3} regions. The refinement found by our system was not generated by Erica due to the additional constraint requiring the output size to be exactly $100$.

\section{Related Work}
\label{sec:related}

\emph{Query refinements.} The problem of query refinement has been addressed in previous studies such as~\cite{muslea2005online, KoudasLTV06, MK09, chu1994structured, tran2010conquer, tran2009query}. 
They focus on modifying queries to satisfy cardinality constraints, mostly emphasizing the overall output size rather than specific data groups within the output, and does not consider ranking of the results. 
For example, \cite{muslea2005online, KoudasLTV06} aim to relax queries with an empty result set to produce some answers. Other works like~\cite{MK09, chu1994structured} address the issues of too many or too few answers by refining queries to meet specific cardinality constraints on the result's size.
A recent line of work studied the use of refinement to satisfy diversity constraints~\cite{MLJ22,ERICA,SSAD22}. The work of~\cite{SSAD22} aims to refine queries to satisfy constraints on the size of specific data groups in the result, however, they consider only numerical predicates with a single binary sensitive attribute. Closer to our work,~Erica \cite{MLJ22,ERICA,ERICAfull} utilizes provenance annotations to efficiently find minimal refinements. While our proposed solution is inspired by these works, their focus is on selection queries and can not be easily extended to ranking queries. Particularly, the provenance model used in these works is insufficiently expressive to capture the semantics of ranking, motivating our need to devise a new way to annotate and use these annotations to find the best approximation refinement. We discuss and demonstrate the differences
in \Cref{sec:erica_comparison}.

\emph{Constrained query answering.} More generally, our problem answers queries that are subject to some set of constraints over the results. Systems like those proposed in~\cite{BRAM15,BAM14} allow querying groups of tuples that optimize some objective function while satisfying some constraints on the output, including cardinality constraints. However, they do not support top-$k$ queries and therefore do not extend to the ranking setting. The work in~\cite{BAM14} specifically relaxes the constraints of the problem to achieve partial satisfaction of the set of constraints, however it does so by removing constraints and not by modifying them as in our work. In \cite{Tiresias}, the authors develop a system to answer how-to queries. How-to queries answer how to modify the database in order to satisfy some constraints while optimizing for an objective. However, their system also lacks support for ranking, making it unsuitable to use for intervening on the top-$k$ for various $k$ values as in our framework.

\emph{Fairness in ranking.} The problem we consider in this paper have implications in the context of fairness. Fairness in ranking has been the subject of much recent attention~\cite{ZYS23, ZYS23b, YGS19,CSV18, AJS19, IWSR22, CMR23, YS17, KR18, CMV20}. These works can be categorized as post-processing methods (e.g.,~\cite{YS17, YGS19,CSV18}) that directly modify the output rankings, or in-processing solutions~\cite{AJS19, IWSR22, CMR23, KR18, CMV20} that adjust the ranking algorithm or modify items to produce a different score. Our solution can be considered as an in-processing method, however unlike existing solutions, we assume ranking algorithms and scores of different items are well-designed, and do not modify them.

\emph{Query result diversification.} 
Query result diversification aims to increase result diversity while maintaining relevance of results to the original query by including or excluding tuples from the set of tuples in the result of the query output.~\cite{GS09,VRB+11,DF14}. Unlike our solution, the diversification is achieved by modifying the set of the tuples directly rather than the query, and does not consider tuples absent from the original query.

\emph{MILP \& databases.} Mixed-integer linear programming has been used in data management in order to solve relevant {\sf NP-hard} optimization problems. However, as pointed out in \cite{Tiresias,QFix,BAM18,PackageQueriesScale}, scaling MILP problems to database-size problems is difficult. In order to scale, these works make several optimizations. In particular, the relevancy-based optimization we proposed resembles optimizations presented in~\cite{QFix,Tiresias}.

\section{Conclusion}
\label{sec:conclusion}

We identified a novel intervention to diversify (according to user-input constraints) the output of top-$k$ queries by refining the selection predicates of the input query. Furthermore, we recognized the importance of maintaining the user's intent as best as possible when searching for such a refinement. Towards this end, we developed a framework that can find the closest refinement for various distance measures that satisfy the user's desired constraints. 
We introduced optimizations in order to make our framework practical for datasets of real-life scale. We demonstrated this with a suite of experiments, showing our framework's scaling capability and the usefulness of our optimizations. In the future, this problem could be extended to find refinements that remain diverse even after adding new data. This way, the refinement may explain some underlying bias of the query instead of fitting to the original data. Extending our model to richer classes of queries presents further interesting directions as we discussed in \Cref{sec:search}.

\bibliographystyle{ACM-Reference-Format}

\end{document}